\documentclass[11pt]{article}
\usepackage{natbib}
\usepackage{amsmath,amssymb,amsthm}
\usepackage{graphicx}
\usepackage{booktabs}
\usepackage[margin=1in]{geometry}
\usepackage{setspace}
\newtheorem{proposition}{Proposition}
\newtheorem{corollary}{Corollary}
\newtheorem{remark}{Remark}

\begin{document}

\title{Bayesian Dynamic Gamma Models for Route-Level Travel Time Reliability}
\author{
Vadim~Sokolov\thanks{Department of Systems Engineering and Operations Research, George Mason University. Email: \texttt{vsokolov@gmu.edu}} \and
Refik~Soyer\thanks{Department of Decision Sciences, George Washington University. Email: \texttt{rsoyer@gwu.edu}}
}
\date{}
\maketitle

\begin{abstract}
Route-level travel time reliability requires characterizing the distribution of total travel time across correlated segments---a problem where existing methods either assume independence (fast but miscalibrated) or model dependence via copulas and simulation (accurate but expensive). We propose a conjugate Bayesian dynamic Gamma model with a common random environment that resolves this trade-off. Each segment's travel time follows a Gamma distribution conditional on a shared latent environment process that evolves as a Markov chain, inducing cross-segment dependence while preserving conditional independence. A moment-matching approximation yields a closed-form $F$-distribution for route travel time, from which the Planning Time Index, Buffer Index, and on-time probability are computed instantly---at the same $O(1)$ cost as independence-based methods. The conjugate structure ensures that Bayesian posterior updates and the full predictive distribution are available in closed form as new sensor data arrives. Applied to 16 sensors spanning 8.26 miles on I-55 in Chicago, the model achieves 95.4\% coverage of nominal 90\% predictive intervals versus 34--37\% for independence-based convolution, at identical computational cost.
\end{abstract}

\section{Introduction}\label{sec:intro}

Travel time reliability---the variability and predictability of travel times---is a primary concern for both travelers and transportation system operators. Reliability is distinct from average travel time: two routes with the same mean travel time may differ dramatically in their day-to-day variability, and research has consistently shown that travelers value reliability as much as, or more than, average travel time \citep{liu04,tal10}. The U.S.\ Federal Highway Administration has identified reliability as a key performance measure, noting that it is ``more important and meaningful to the public than other congestion measures such as level of service, speed, and volume'' \citep{FHWA2006}.

Quantifying travel time reliability requires moving beyond first-order moments to characterize the full distribution of travel times. This is particularly challenging at the route level, where the total travel time is the sum of correlated segment travel times, each with a non-standard distribution. The correlated link-to-route aggregation problem is a central methodological challenge \citep{FHWA2024}: existing methods that assume segment independence are computationally fast but dramatically underestimate route variance, while methods that model dependence (e.g., copulas) require expensive Monte Carlo simulation. The central question addressed in this paper is: \emph{how can we compute the predictive distribution of route-level travel time in a way that (i) accounts for cross-segment dependence, (ii) captures the dynamic, time-varying nature of traffic conditions, and (iii) yields closed-form reliability metrics suitable for real-time operations?}

We propose a dynamic Gamma model with a common random environment that simultaneously addresses all three requirements. The key modeling insight is that cross-segment dependence in travel times is driven primarily by shared conditions---congestion waves, weather, incidents---that affect an entire corridor. By representing these shared conditions through a single latent process $\eta_t$, the model induces realistic positive dependence across segments while maintaining conditional independence. This structure has a powerful computational consequence: the route travel time distribution reduces from an $m$-dimensional integration problem to a single one-dimensional integral, regardless of the number of segments. Combined with a Gamma moment-matching approximation, this yields a closed-form $F$-distribution for route travel time, from which all standard reliability metrics can be computed in $O(1)$ time---the same computational cost as methods that simply assume independence, but with correct accounting for cross-segment dependence.

Our main contributions are:
\begin{enumerate}
\item A multivariate dynamic Gamma model where a common random environment process induces cross-segment dependence while preserving analytical tractability. The model nests the multivariate Lomax distribution of \citet{Nayak87} and the dynamic extension of \citet{APS20} as special cases.
\item A moment-matching approximation that reduces the conditional distribution of route travel time (a sum of heterogeneous Gamma variates) to a single Gamma, preserving the conjugate structure and yielding the route predictive distribution in closed form via the $F$-distribution. This resolves the correlated link-to-route aggregation problem at the same $O(1)$ computational cost as independence-based convolution methods.
\item Closed-form, time-varying expressions for the Planning Time Index, Buffer Index, and on-time arrival probability that update sequentially as new data arrives.
\item An empirical comparison demonstrating that independence-based methods achieve only 34--37\% coverage of nominal 90\% predictive intervals due to variance underestimation, while our model achieves 95.4\% coverage at identical computational cost, and outperforms copula-based simulation in calibration.
\end{enumerate}

The rest of the paper is organized as follows. Section~\ref{sec:literature} reviews the literature on travel time distributions, route-level reliability computation, and dynamic Bayesian models in transportation. Section~\ref{sec:model} presents the dynamic Gamma model for a single segment. Section~\ref{sec:multivariate} develops the multivariate extension with common random environment. Section~\ref{sec:route} derives the main theoretical result: closed-form route-level reliability via moment matching. Section~\ref{sec:inference} discusses Bayesian inference. Section~\ref{sec:application} presents the empirical application to I-55 traffic data. Section~\ref{sec:conclusion} concludes with discussion and extensions.

\section{Literature Review}\label{sec:literature}

\subsection{Travel Time Distribution Modeling}

A prerequisite for reliability analysis is an appropriate distributional model for travel times. Early work assumed normally distributed travel times \citep{iida99}, which yields tractable route-level convolutions but fails to capture the positive support, right-skewness, and heavy tails that characterize empirical travel time data.

The empirical literature has documented that travel times are better described by skewed, positive-valued distributions. \citet{Susilawati2013} compare several parametric families---Gamma, lognormal, Burr, and log-logistic---for urban road travel times, finding that the lognormal and Burr distributions provide the best fits for longer links while the Gamma is competitive for shorter segments. \citet{Fosgerau2012} model travel time on an urban road using a lognormal distribution and relate its parameters to the value of travel time variability. The comprehensive review by \citet{Zang2022} catalogs the range of distributions that have been applied and emphasizes that no single parametric family dominates across all settings.

An important observation from this literature is the distinction between \emph{static} and \emph{dynamic} distributional models. Most existing work fits a single distribution to a pooled sample of travel times, implicitly assuming stationarity. However, travel time distributions evolve within a day (e.g., free-flow in the morning versus congestion in the afternoon) and across days (e.g., in response to incidents or weather). \citet{vanLint2008} argue that variance-based measures of unreliability are insufficient because the same variance can arise from very different distributional shapes, and advocate for characterizing the full time-varying distribution. The recent comprehensive review by \citet{ZangXu2022} reinforces this point, organizing the literature into a four-stage framework---characterization, evaluation, valuation, and assignment---and identifying the need for models that simultaneously handle non-stationarity and cross-segment dependence.

A further critical distinction is between models that operate at the \emph{segment level} versus those that model route-level travel time directly. Segment-level models are more fundamental: each road segment has distinct physical characteristics (length, capacity, merge/diverge geometry) that produce segment-specific distributional shapes. A model that captures these segment-level dynamics can be composed to produce route-level predictions for \emph{any} route, whereas a model fit directly to route data is tied to a single origin--destination pair. Our approach models each segment individually through a Gamma distribution with segment-specific rate parameter $\lambda_j$, and builds route-level reliability from these components.

\subsection{Route-Level Reliability and the Convolution Problem}

Computing the distribution of route travel time $S_t = \sum_{j=1}^m y_{jt}$ requires the convolution of the segment-level distributions. Under the assumption of normally distributed, independent segment travel times, the convolution is trivially normal \citep{iida99}. For more realistic distributions, the convolution does not have a closed form.

\citet{ng10} address this by evaluating the convolution integral numerically using the Fast Fourier Transform (FFT), applicable to any distribution with a known characteristic function. \citet{ng11} develop a distribution-free approach based on probability inequalities that avoids specifying a parametric family altogether but produces conservative bounds rather than exact distributions. \citet{nie12} apply FFT-based convolution in a practical route guidance system for the Chicago expressway network. \citet{Lei2014} model travel times on an urban expressway by fitting GEV and GPD distributions to segment-level travel time per unit distance and then convolving these distributions to obtain route-level reliability under varying levels of service. \citet{Ramezani2012} take an alternative approach, using Markov chains to model the state transitions of each link and assembling the path travel time distribution through convolution of link-specific distributions.

A recent FHWA methodology report \citep{FHWA2024} on estimating travel time distributions along user-defined paths highlights that constructing path distributions from link-level data while preserving correlations remains a central methodological challenge, particularly as agencies move from fixed-point sensor data to trajectory-based probe data.

All of these convolution-based approaches assume that segment travel times are \emph{independent}---an assumption that is widely acknowledged to be unrealistic. Adjacent segments on a route share common traffic conditions, and empirical evidence shows strong positive correlation between consecutive segment travel times \citep{Zockaie2014}. Ignoring this dependence understates the variance of route travel time and produces unreliable predictive intervals. \citet{Filipovska2021} address this by estimating path travel time distributions in stochastic, time-varying networks with generalized spatio-temporal correlations, using Monte Carlo simulation to convolve correlated link travel times and adaptively update route distributions en route. More recently, \citet{Kamphuis2025} develop a nonparametric estimator for the joint distribution of per-edge travel times in spatially dependent networks, with provable consistency guarantees; however, their Gaussian framework does not capture the skewness inherent in travel time data.

Copula-based models offer a natural framework for specifying cross-segment dependence separately from the marginal distributions. \citet{ChenYu2017} develop a copula-based approach for estimating travel time reliability on an urban arterial, fitting marginal distributions to individual segment travel times and then using copulas to model the cross-segment dependence structure. Their approach correctly identifies that each segment requires its own distributional model and that cross-segment dependence matters; however, computing the distribution of the route sum $S_t$ under their copula model requires $m$-dimensional numerical integration or Monte Carlo simulation, which scales poorly with the number of segments.

At the network level, \citet{Saedi2020} use network partitioning to decompose a large heterogeneous network into homogeneous sub-networks, estimating travel time reliability within each partition using the linear mean--standard deviation relationship. Their approach offers a practical way to handle network heterogeneity but relies on simulation-based estimation and does not provide a closed-form route-level distribution.

\citet{Westgate2016} develop a large-network travel time distribution model for ambulance routing that uses a log-Gaussian Cox process to capture spatial and temporal dependence, but the model requires simulation-based inference that is computationally intensive. Simulation-based approaches are effective for offline analyses but are poorly suited for real-time reliability computation where metrics must be updated as each new data point arrives.

Our approach resolves the convolution problem through a different mechanism: instead of modeling dependence through a copula and then integrating over the joint distribution, we use a common random environment that induces dependence \emph{implicitly}. The resulting conditional independence structure means the route distribution requires only a one-dimensional integral over the environment, regardless of the number of segments. Unlike the convolution approaches of \citet{Lei2014} and \citet{Ramezani2012}, our method naturally accommodates cross-segment dependence; unlike the copula approach of \citet{ChenYu2017}, it avoids $m$-dimensional integration.

\subsection{Dynamic and Bayesian Models for Traffic}

Dynamic models that allow parameters to evolve over time are natural for traffic data, where conditions change rapidly. The Bayesian dynamic linear model framework of \citet{WestHarrison1997} provides a general approach, and \citet{Tebaldi2002} apply it to freeway traffic flow data from the Research Triangle region of North Carolina.

A prominent approach to capturing time-varying reliability is the GARCH (Generalized Autoregressive Conditional Heteroscedasticity) family, adapted from financial econometrics. \citet{Yang2018} develop an ARIMA-GARCH model that produces time-varying confidence intervals for arterial travel time, with intervals that widen during congestion and narrow during free-flow. The GARCH framework explicitly models volatility clustering---the phenomenon that high-variance periods tend to follow other high-variance periods---which is a natural feature of traffic dynamics. However, GARCH models operate on a single aggregate time series (typically route-level) and do not provide segment-level distributional models or a mechanism for composing route distributions from components. In contrast, our discount factor $\gamma$ plays an analogous role to the GARCH persistence parameter, but operates within a conjugate Bayesian framework that yields closed-form posterior updates and a full predictive distribution, not just confidence intervals.

Particle filtering methods have been applied to traffic state estimation by \citet{Mihaylova2007}, who develop a particle filter for freeway traffic flow within a macroscopic traffic model. \citet{PolsonSokolov2015} develop Bayesian particle tracking methods for traffic flows that combine physical traffic models with statistical estimation. The particle learning framework of \citet{Nick2010} provides a general methodology for sequential Bayesian updating with sufficient statistics that is directly applicable to our model. \citet{KimYe2022} develop a Bayesian mixture model for freeway travel time estimation under sparse probe data, using data-driven priors from historical neighbors; their Gibbs sampling approach yields credible intervals but requires MCMC computation and does not extend to route-level aggregation.

More recently, deep learning methods have achieved strong performance on traffic prediction tasks. \citet{PolsonSokolov2017} demonstrate that deep neural networks can capture complex non-linear patterns in short-term traffic flow prediction. However, deep learning approaches typically produce point predictions or require additional calibration to produce reliable uncertainty estimates. They do not directly yield the \emph{distributional} outputs (CDFs, quantiles, predictive intervals) that reliability analysis requires.

Our model occupies a distinct position in this landscape: for \emph{fixed} hyperparameters $(\alpha, \gamma, \lambda)$, it is a conjugate Bayesian model that provides exact sequential updating of the environment $\eta_t$ and yields the full predictive distribution in closed form. When the hyperparameters are themselves unknown, particle learning methods can be used for joint online estimation (Section~\ref{sec:inference}). The conjugacy is preserved even at the route level through the moment-matching approximation. As noted in the recent review by \citet{ZangXu2022}, the gap between distributional modeling and real-time computation remains a key challenge for the field; our closed-form approach directly addresses this gap.

\subsection{Common Random Environment Models}

The idea of using a shared latent variable to induce dependence among conditionally independent components has a long history in reliability theory. \citet{LS86} introduce the concept of a ``common environment'' shared by multiple components, showing that a Gamma-distributed random environment induces multivariate Lomax distributions for the component lifetimes. \citet{Nayak87} develops the statistical properties of the resulting multivariate Lomax distribution, including its moments and dependence structure. \citet{Singpurwalla1995} extends the framework to dynamic environments where the operating conditions change over time, and \citet{Hougaard2000} provides a comprehensive treatment of related frailty models in survival analysis.

\citet{ASX13} extend the Lindley--Singpurwalla framework to a dynamic Bayesian setting where the latent environment follows a beta-driven Markov process, applying the model to mortgage default counts. \citet{APS20} develop a general family of multivariate non-Gaussian time series models that nests the Gamma, Poisson, and Binomial cases; the Gamma member of this family yields the dynamic multivariate beta prime distribution that we adopt here. \citet{APS20} also provide the FFBS backward sampling and MCMC inference framework used in Section~\ref{sec:inference}. \citet{SAK15} use the common random environment structure to model dependence in count data time series.

Our paper extends this line of work in a specific and practically important direction: we show that the common environment structure, combined with a moment-matching approximation for the route-level sum, yields a \emph{closed-form} route travel time distribution. This result transforms the common random environment model from a framework for segment-level modeling into a complete tool for route-level reliability analysis.

\subsection{Research Gap and Contribution}

Table~\ref{tab:positioning} summarizes how existing approaches address the key requirements for route-level reliability analysis. Convolution-based methods \citep{ng10,Lei2014,Ramezani2012} build route distributions from segment models but assume independence and use static distributions. The copula approach of \citet{ChenYu2017} correctly models segment-level distributions and cross-segment dependence, but requires $m$-dimensional simulation to compute route reliability and does not accommodate dynamic updating. Network-level methods \citep{Saedi2020} bypass segment-level modeling entirely, sacrificing the ability to compose reliability estimates for arbitrary routes. No prior method simultaneously provides (i) segment-level distributional models, (ii) dynamic updating that tracks time-varying conditions, (iii) a multivariate structure that captures cross-segment dependence, and (iv) closed-form route-level reliability metrics. Our framework fills this gap.

\begin{table}[htbp]
\centering
\caption{Positioning of the proposed framework relative to existing approaches for route-level travel time reliability.}
\label{tab:positioning}
\small
\begin{tabular}{lccc}
\toprule
Approach & Dynamic & Dependence & Closed-form \\
\midrule
Normal convolution \citep{iida99} & No & No & Yes \\
FFT convolution \citep{ng10} & No & No & Numerical \\
GEV/GPD convolution \citep{Lei2014} & No & No & Numerical \\
Markov chain \citep{Ramezani2012} & No & No & Numerical \\
Copula \citep{ChenYu2017} & No & Yes & Simulation \\
Correlated MCS \citep{Filipovska2021} & Yes & Yes & Simulation \\
Network partition \citep{Saedi2020} & No & Implicit & Simulation \\
ARIMA-GARCH \citep{Yang2018} & Yes & No & Intervals \\
Deep learning \citep{PolsonSokolov2017} & Yes & Implicit & No \\
Simulation \citep{Westgate2016} & Yes & Yes & Simulation \\
\textbf{This paper} & \textbf{Yes} & \textbf{Yes} & \textbf{Yes} \\
\bottomrule
\end{tabular}
\end{table}

\section{Dynamic Gamma Model}\label{sec:model}

\subsection{Observation Model}

There are multiple sources of uncertainty about travel times on a given road segment. These include changes in demand due to daily variations or special events \citep{clark2005modelling}, lane closures due to accidents and road works \citep{chen2003travel}, and degradation due to weather \citep{tu2007impact}. We decompose these into observable covariates and a latent environment component.

We assume that travel time $y_t$ on a road segment at period $t$ follows a Gamma distribution with shape $\alpha$ and rate $\eta_t e^{\beta^T u_t}$, where $u_t$ is an observable covariate vector representing environmental variables and $\eta_t$ is the latent environment at time $t$:
\begin{equation}
y_t \mid \alpha,\eta_t,\beta,u_t \sim \mathrm{Gam}(\alpha,\, \eta_t e^{\beta^T u_t}).
\label{eq:obsmodel}
\end{equation}
Here we use the rate parameterization, so the density is
\[
p(y_t \mid \alpha,\eta_t,\beta,u_t) = \frac{(\eta_t e^{\beta^T u_t})^\alpha}{\Gamma(\alpha)}\, y_t^{\alpha-1}\, e^{-\eta_t e^{\beta^T u_t} y_t}.
\]
The mean travel time is $E[y_t \mid \eta_t] = \alpha / (\eta_t e^{\beta^T u_t})$, which is inversely proportional to the environment quality $\eta_t$: higher values of $\eta_t$ correspond to better conditions (lower travel times), while lower values indicate congestion.

The Gamma distribution is a natural choice for travel time modeling. It is supported on $(0,\infty)$, accommodates right-skewness, and nests the exponential distribution ($\alpha = 1$) as a special case. Empirical studies have confirmed that the Gamma provides a good fit to segment-level travel time data \citep{Susilawati2013,Zang2022}.

Conditional on $\alpha$ and $\eta_t e^{\beta^T u_t}$, travel times at different periods are assumed independent. We define $D^{t-1} = (y^{t-1}, u^{t-1})$ where $y^{t-1} = (y_1,\ldots,y_{t-1})$ and $u^{t-1} = (u_1,\ldots,u_{t-1})$.

\subsection{Dynamic Random Environment}\label{sec:environment}

The latent environment $\eta_t$ evolves according to a Markovian process. Given the history $D^{t-1}$, we specify
\begin{equation}
\eta_t = \frac{\eta_{t-1}}{\gamma}\,\epsilon_t,
\label{eq:markov}
\end{equation}
where $\epsilon_t \mid \alpha,\gamma,\beta,D^{t-1} \sim \mathrm{Beta}[\gamma a_{t-1},\, (1-\gamma)a_{t-1}]$ and $0 < \gamma < 1$ is a discount factor. This evolution has its origins in the dynamic environment models of \citet{Singpurwalla1995} and has been used in reliability growth testing \citep{LS86} and in modeling mortgage default rates \citep{ASX13}.

The discount factor $\gamma$ controls how quickly the model adapts to new information. Small values of $\gamma$ (rapid discounting) make the model responsive to recent observations, while values close to 1 give more weight to historical data. The evolution implies a stochastic ordering $\eta_t < \eta_{t-1}/\gamma$ and has a random walk property:
\[
E(\eta_t \mid \eta_{t-1}, D^{t-1}) = \eta_{t-1},
\]
where $(\eta_t \mid \eta_{t-1}, D^{t-1})$ is a scaled Beta density over $(0, \eta_{t-1}/\gamma)$.

\subsection{Sequential Bayesian Updating}

A key feature of the model is that the posterior distribution of $\eta_t$ remains in the Gamma family, enabling exact sequential updating.

\begin{proposition}\label{prop:updating}
If $(\eta_0 \mid D^0) \sim \mathrm{Gam}(a_0, b_0)$, then for all $t \geq 1$:
\begin{enumerate}
\item[(i)] The posterior at time $t-1$ is $(\eta_{t-1} \mid D^{t-1}) \sim \mathrm{Gam}(a_{t-1}, b_{t-1})$.
\item[(ii)] The prior at time $t$ is $(\eta_t \mid D^{t-1}) \sim \mathrm{Gam}(\gamma a_{t-1}, \gamma b_{t-1})$.
\item[(iii)] After observing $y_t$, the posterior updates to $(\eta_t \mid D^t) \sim \mathrm{Gam}(a_t, b_t)$ where
\begin{equation}
a_t = \gamma a_{t-1} + \alpha, \qquad b_t = \gamma b_{t-1} + y_t e^{\beta^T u_t}.
\label{eq:update}
\end{equation}
\end{enumerate}
\end{proposition}
\begin{proof}
Part~(ii) follows from the Markov evolution~\eqref{eq:markov} \citep{ASX13}. For~(iii), the Gamma likelihood is conjugate to the Gamma prior in~(ii); multiplying the kernels gives shape $\gamma a_{t-1} + \alpha$ and rate $\gamma b_{t-1} + y_t e^{\beta^T u_t}$.
\end{proof}

The updating is immediate: the conjugate Gamma prior for $\eta_t$ combined with the Gamma likelihood yields a Gamma posterior. The discount factor $\gamma$ systematically inflates the variance of the prior relative to the posterior---statement (ii) shows that the prior variance $\gamma a_{t-1} / (\gamma b_{t-1})^2 = a_{t-1}/(\gamma b_{t-1}^2)$ exceeds the posterior variance $a_{t-1}/b_{t-1}^2$ by a factor of $1/\gamma$.

\begin{corollary}[Stochastic properties]
The environment process satisfies:
\[
E(\eta_t \mid D^{t-1}) = E(\eta_{t-1} \mid D^{t-1}), \qquad V(\eta_t \mid D^{t-1}) > V(\eta_{t-1} \mid D^{t-1}).
\]
That is, the predictive mean of the environment is preserved, but uncertainty increases between updates.
\end{corollary}

\subsection{Predictive Distribution}

The one-step-ahead predictive distribution of $y_t$ is obtained by integrating out the environment:
\begin{equation}
p(y_t \mid \alpha,\gamma,\beta,D^{t-1}) = \int_0^\infty p(y_t \mid \alpha,\eta_t,\beta,u_t)\,p(\eta_t \mid D^{t-1})\,d\eta_t.
\label{eq:predictive_integral}
\end{equation}
This yields a compound Gamma distribution:
\begin{equation}
p(y_t \mid D^{t-1}) = \frac{\Gamma(\alpha + \gamma a_{t-1})}{\Gamma(\alpha)\,\Gamma(\gamma a_{t-1})}\, \frac{(e^{\beta^T u_t})^{\alpha}\,y_t^{\alpha-1}\,(\gamma b_{t-1})^{\gamma a_{t-1}}}{(\gamma b_{t-1} + y_t e^{\beta^T u_t})^{\alpha + \gamma a_{t-1}}},
\label{eq:predictive}
\end{equation}
which is a Beta prime (compound Gamma) distribution, reducing to the Lomax (Pareto Type~II) distribution when $\alpha = 1$. An important consequence is that the predictive quantiles are available in closed form via the $F$-distribution: if we write $\tilde{a} = \gamma a_{t-1}$ and $\tilde{b} = \gamma b_{t-1}$, then
\begin{equation}
\frac{\tilde{a}}{\alpha}\cdot\frac{y_t e^{\beta^T u_t}}{\tilde{b}} \;\sim\; F(2\alpha,\, 2\tilde{a}).
\label{eq:F_relation}
\end{equation}
This means that the predictive CDF, quantiles, and moments can all be computed using standard $F$-distribution functions---no numerical integration or simulation is required.

\begin{remark}
Suppressing covariates (i.e., setting $\beta = 0$) for notational clarity, the predictive mean and variance are
\[
E[y_t \mid D^{t-1}] = \frac{\alpha\,\tilde{b}}{\tilde{a} - 1}\quad (\tilde{a}>1), \qquad
V[y_t \mid D^{t-1}] = \frac{\alpha\,\tilde{b}^2(\alpha + \tilde{a}-1)}{(\tilde{a}-1)^2(\tilde{a}-2)}\quad (\tilde{a}>2).
\]
When covariates are present, both expressions are scaled by $e^{-\beta^T u_t}$ and $e^{-2\beta^T u_t}$, respectively. The heavier tails relative to the Gamma distribution (the predictive variance is inflated by the factor $(\alpha + \tilde{a}-1)/(\tilde{a}-2)$) reflect the additional uncertainty about the environment.
\end{remark}

\section{Multivariate Extension: Common Random Environment}\label{sec:multivariate}

\subsection{Model Specification}

Consider a route composed of $m$ road segments. The travel time on segment $j$ at time $t$ follows
\begin{equation}
y_{jt} \mid \lambda_j,\alpha,\eta_t,\beta,u_t \sim \mathrm{Gam}(\alpha,\, \lambda_j \eta_t e^{\beta^T u_t}),
\label{eq:multivariate}
\end{equation}
where $\lambda_j > 0$ is a segment-specific parameter capturing the relative speed of segment $j$, and $\eta_t$ is the common random environment shared across all segments. Segments with higher $\lambda_j$ have lower expected travel times: $E[y_{jt} \mid \eta_t] = \alpha / (\lambda_j \eta_t e^{\beta^T u_t})$.

The critical modeling assumption is \textbf{conditional independence}: given $\eta_t$ and the $\lambda_j$'s, the segment travel times $y_{1t}, \ldots, y_{mt}$ are independent. The common environment $\eta_t$ induces all cross-segment dependence. This is physically motivated: the dominant source of correlation between adjacent segments is shared network conditions (congestion waves, weather, incidents), which are exactly what $\eta_t$ captures.

This dependence structure, where a common latent variable induces correlation among conditionally independent observations, has its origins in the shared environment models of \citet{LS86} and \citet{Nayak87}, and has been applied to credit risk modeling \citep{ASX13} and count data time series \citep{SAK15}.

An important structural consequence is that the predictive correlation between any two segments $i \neq j$ takes the form $\mathrm{Corr}(y_{it}, y_{jt} \mid D^{t-1}) = \alpha / (\alpha + \tilde{a} - 1)$, which depends on $\alpha$ and $\tilde{a}$ but not on $i$ or $j$. That is, the model implies \emph{equicorrelation}: all segment pairs share the same correlation, regardless of their spatial separation. Empirically, adjacent segments tend to exhibit higher correlation than distant ones (Section~\ref{sec:application}). Nevertheless, the model captures the dominant first-order effect---pervasive positive dependence driven by shared conditions---and the route-level aggregation (which sums over all pairs) is governed by the average correlation rather than its spatial structure. We revisit this limitation in Section~\ref{sec:conclusion}.

\subsection{Sequential Updating}

All results from the univariate model carry over with modified recursions. The posterior of $\eta_t$ given all $m$ segments updates as
\begin{equation}
(\eta_t \mid \lambda, \alpha, \gamma, \beta, D^t) \sim \mathrm{Gam}(a_t, b_t),
\label{eq:mv_posterior}
\end{equation}
where
\[
a_t = \gamma a_{t-1} + m\alpha, \qquad b_t = \gamma b_{t-1} + e^{\beta^T u_t}\sum_{j=1}^m \lambda_j y_{jt}.
\]
The common environment learns simultaneously from all $m$ segments, weighted by their $\lambda_j$ values. With $m$ segments each contributing $\alpha$ to the shape parameter, the effective information per time step is $m\alpha$, leading to tighter posteriors and more precise reliability estimates.

\begin{remark}[Identifiability]
The parameters $\lambda_j$ and $\eta_t$ enter the likelihood only through their product $\lambda_j \eta_t$: rescaling $\lambda_j \to c\lambda_j$ and $\eta_t \to \eta_t/c$ leaves the likelihood invariant. In the Bayesian framework, the Gamma prior $(\eta_0 \mid D^0) \sim \mathrm{Gam}(a_0, b_0)$ anchors the scale of $\eta_t$; the $\lambda_j$'s then represent \emph{relative} segment rates. In practice, we normalize $\lambda_j$ so that $\bar\lambda = 1$, resolving the scale ambiguity (see Section~\ref{sec:application}).
\end{remark}

\subsection{Joint Predictive Distribution}

The joint predictive distribution of segment travel times is obtained by integrating out the environment. For notation simplicity, we suppress the covariates $u_t$ and write $\tilde{a} = \gamma a_{t-1}$, $\tilde{b} = \gamma b_{t-1}$.

\begin{proposition}[Multivariate compound Gamma]\label{prop:joint}
The joint predictive density is
\begin{equation}
p(y_{1t},\ldots,y_{mt} \mid D^{t-1}) = \frac{\Gamma(m\alpha + \tilde{a})}{\Gamma(\alpha)^m\,\Gamma(\tilde{a})}\,
\frac{\prod_{j=1}^m \lambda_j^\alpha\, y_{jt}^{\alpha-1}\;\cdot\; \tilde{b}^{\tilde{a}}}
{\left(\tilde{b} + \sum_{j=1}^m \lambda_j y_{jt}\right)^{m\alpha + \tilde{a}}}.
\label{eq:joint_predictive}
\end{equation}
\end{proposition}
\begin{proof}
The integral $\int_0^\infty \prod_j p(y_{jt} \mid \eta_t) \cdot p(\eta_t \mid D^{t-1})\,d\eta_t$ yields the product of Gamma densities times the Gamma prior, integrated as a kernel of $\mathrm{Gam}(m\alpha + \tilde{a},\, \tilde{b} + \sum_j \lambda_j y_{jt})$.
\end{proof}

When $\alpha = 1$ (exponential observation model), equation~\eqref{eq:joint_predictive} reduces to the dynamic multivariate Lomax distribution studied by \citet{Nayak87} in the static case and by \citet{APS20} in the dynamic setting. For general $\alpha$, it is a multivariate compound Gamma distribution.

\section{Route-Level Travel Time Reliability}\label{sec:route}

This section presents the main theoretical contribution of the paper: a closed-form route predictive distribution that enables instant computation of reliability metrics.

\subsection{Dimension Reduction via Conditional Independence}

The route travel time is $S_t = \sum_{j=1}^m y_{jt}$. The predictive distribution of $S_t$ is
\begin{equation}
P(S_t \leq \tau \mid D^{t-1}) = \int_0^\infty P(S_t \leq \tau \mid \eta_t)\, p(\eta_t \mid D^{t-1})\,d\eta_t.
\label{eq:route_integral}
\end{equation}
The conditional independence structure means that $P(S_t \leq \tau \mid \eta_t)$ involves the distribution of a sum of \emph{independent} Gamma variates---this is a one-dimensional CDF that can be computed efficiently. The outer integral over $\eta_t$ is also one-dimensional.

The computational advantage is evident when compared with alternatives. If segments are independent but non-normally distributed, computing $P(S_t \leq \tau)$ requires $m$-fold convolution of the marginal distributions \citep{ng10}, evaluated numerically via FFT. Copula models fare worse: $P(S_t \leq \tau)$ requires $m$-dimensional numerical integration or simulation. By contrast, our common environment formulation reduces the problem to a single one-dimensional integral over $\eta_t$, regardless of $m$.

\subsection{Gamma Moment-Matching Approximation}

Conditional on $\eta_t$, the segment travel times are independent with $y_{jt} \mid \eta_t \sim \mathrm{Gam}(\alpha, \lambda_j \eta_t)$. The conditional sum $S_t \mid \eta_t$ has moments
\[
E[S_t \mid \eta_t] = \frac{\alpha}{\eta_t}\sum_{j=1}^m \frac{1}{\lambda_j}, \qquad
V[S_t \mid \eta_t] = \frac{\alpha}{\eta_t^2}\sum_{j=1}^m \frac{1}{\lambda_j^2}.
\]
We approximate $S_t \mid \eta_t$ by a single Gamma distribution matched to these moments:
\begin{equation}
S_t \mid \eta_t \;\approx\; \mathrm{Gam}(\alpha^*,\, c\,\eta_t),
\label{eq:moment_match}
\end{equation}
where
\begin{equation}
\alpha^* = \alpha \cdot \frac{\left(\sum_{j=1}^m \lambda_j^{-1}\right)^2}{\sum_{j=1}^m \lambda_j^{-2}}, \qquad
c = \frac{\sum_{j=1}^m \lambda_j^{-1}}{\sum_{j=1}^m \lambda_j^{-2}}.
\label{eq:alpha_star}
\end{equation}

Two important properties of this approximation:
\begin{itemize}
\item The effective route shape $\alpha^*$ does not depend on $\eta_t$. It depends only on the segment parameters $\lambda_j$ and the per-segment shape $\alpha$.
\item The effective route rate $c\,\eta_t$ is linear in $\eta_t$.
\end{itemize}
These properties are crucial because they preserve the conjugate structure of the model, enabling closed-form integration over $\eta_t$.

\begin{remark}
When all $\lambda_j$ are equal (homogeneous segments), $\alpha^* = m\alpha$ and $c = \lambda$, so the approximation is exact: the sum of i.i.d.\ $\mathrm{Gam}(\alpha, \lambda\eta_t)$ variables is exactly $\mathrm{Gam}(m\alpha, \lambda\eta_t)$.
More generally, the Gamma approximation to a sum of independent Gamma variates is known to be accurate when the individual shape parameters are moderate \citep{MoschouSingle}.
\end{remark}

\subsection{Closed-Form Route Predictive Distribution}

Substituting the moment-matched approximation~\eqref{eq:moment_match} into the route integral~\eqref{eq:route_integral} yields the same compound Gamma structure as the univariate case.

\begin{proposition}[Route predictive distribution]\label{prop:route}
Under the Gamma moment-matching approximation, the predictive distribution of route travel time is
\begin{equation}
p(S_t \mid D^{t-1}) = \frac{\Gamma(\alpha^* + \tilde{a})}{\Gamma(\alpha^*)\,\Gamma(\tilde{a})}\,
\frac{c^{\alpha^*}\, S_t^{\alpha^*-1}\, \tilde{b}^{\tilde{a}}}{(c\,S_t + \tilde{b})^{\alpha^* + \tilde{a}}},
\label{eq:route_predictive}
\end{equation}
where $\tilde{a} = \gamma a_{t-1}$ and $\tilde{b} = \gamma b_{t-1}$ are the prior parameters from the multivariate filter. Equivalently,
\begin{equation}
\frac{\tilde{a}}{\alpha^*}\cdot\frac{c\, S_t}{\tilde{b}} \;\sim\; F(2\alpha^*,\, 2\tilde{a}).
\label{eq:route_F}
\end{equation}
\end{proposition}
\begin{proof}
The integral $\int_0^\infty \mathrm{Gam}(S_t; \alpha^*, c\eta_t) \cdot \mathrm{Gam}(\eta_t; \tilde{a}, \tilde{b})\,d\eta_t$ yields the compound Gamma~\eqref{eq:route_predictive} by the standard conjugacy calculation. The $F$-distribution relationship follows from the well-known connection between the Beta prime distribution and the $F$-distribution.
\end{proof}

This result is the central contribution of the paper. It means that:
\begin{itemize}
\item The route predictive CDF is
\begin{equation}
P(S_t \leq \tau \mid D^{t-1}) = F_{2\alpha^*,\, 2\tilde{a}}\!\left(\frac{\tilde{a}}{\alpha^*}\cdot\frac{c\,\tau}{\tilde{b}}\right),
\label{eq:route_cdf}
\end{equation}
where $F_{\nu_1,\nu_2}$ denotes the CDF of the $F$-distribution with degrees of freedom $\nu_1$ and $\nu_2$.
\item The route predictive quantile at level $q$ is
\begin{equation}
S_t^{(q)} = \frac{\alpha^*\,\tilde{b}}{c\,\tilde{a}}\cdot F^{-1}_{2\alpha^*,\, 2\tilde{a}}(q).
\label{eq:route_quantile}
\end{equation}
\item Both are available in $O(1)$ computation time via standard statistical functions.
\end{itemize}

\subsection{Route Reliability Metrics}

Using the closed-form route predictive, standard reliability metrics recommended by \citet{FHWA2006} and \citet{Lomax2003} can be computed instantly at each time step.

\paragraph{On-time arrival probability.} For a threshold $\tau$ (e.g., the free-flow travel time plus a buffer),
\[
P_{\mathrm{on\text{-}time}}(t) = P(S_t \leq \tau \mid D^{t-1}) = F_{2\alpha^*,2\tilde{a}}\!\left(\frac{\tilde{a}\,c\,\tau}{\alpha^*\,\tilde{b}}\right).
\]

\paragraph{Planning Time Index (PTI).} The ratio of the 95th percentile of travel time to the free-flow travel time:
\[
\mathrm{PTI}(t) = \frac{S_t^{(0.95)}}{S_{\mathrm{ff}}},
\]
where $S_{\mathrm{ff}}$ is the free-flow route travel time and $S_t^{(0.95)}$ is given by~\eqref{eq:route_quantile} with $q=0.95$.

\paragraph{Buffer Index.} The normalized difference between the 95th percentile and the median:
\[
\mathrm{BI}(t) = \frac{S_t^{(0.95)} - S_t^{(0.50)}}{S_t^{(0.50)}}.
\]

All three metrics are functions of $F$-distribution quantiles evaluated at the current prior parameters $(\tilde{a}, \tilde{b})$, making them computable in real time. This contrasts sharply with existing approaches where reliability metrics are either computed from historical samples (static, not real-time) or require simulation \citep{Taylor2013}.

\section{Bayesian Inference}\label{sec:inference}

\subsection{MCMC via Gibbs Sampling and FFBS}

For batch inference from a complete dataset $D^T$, we develop a Gibbs sampler. The joint posterior of all unknown parameters $\Psi = (\lambda, \alpha, \gamma, \beta)$ and the environmental process $\eta^T = (\eta_1,\ldots,\eta_T)$ cannot be obtained in closed form, but the model structure enables efficient blocked sampling.

The Gibbs sampler alternates between:
\begin{enumerate}
\item \textbf{Sampling $\eta^T$ given $\Psi$ and $D^T$}: We use the forward filtering backward sampling (FFBS) algorithm of \citet{FS94}. The forward pass computes the filtering distributions $(\eta_t \mid \Psi, D^t) \sim \mathrm{Gam}(a_t, b_t)$ using the multivariate recursion $a_t = \gamma a_{t-1} + m\alpha$, $b_t = \gamma b_{t-1} + \sum_j \lambda_j y_{jt} e^{\beta^T u_t}$ (cf.~equation~\eqref{eq:update} for the univariate case). The backward pass draws $\eta_T \sim \mathrm{Gam}(a_T, b_T)$ and then, for $t = T-1, \ldots, 1$, uses the shifted Gamma characterization of \citet{APS20}: draw $z_t \sim \mathrm{Gam}((1-\gamma)a_t, b_t)$ and set $\eta_t = \gamma \eta_{t+1} + z_t$. This avoids truncation-based sampling and is efficient to implement.

\item \textbf{Sampling $\lambda_j$ given $\eta^T$, other parameters, and $D^T$}: With independent Gamma priors $\lambda_j \sim \mathrm{Gam}(r_{j0}, q_{j0})$, the full conditional is
\[
(\lambda_j \mid \eta^T, D^T, \alpha, \beta) \sim \mathrm{Gam}\!\left(r_{j0} + T\alpha,\; q_{j0} + \sum_{t=1}^T \eta_t y_{jt} e^{\beta^T u_t}\right).
\]
The $\lambda_j$'s are conditionally independent given $\eta^T$.

\item \textbf{Sampling $\gamma$}: The discount factor does not have a conjugate full conditional. Following \citet{APS20}, we use a Metropolis step targeting the marginal posterior $p(\gamma \mid D^T, \alpha, \beta, \lambda) \propto \prod_t p(\mathbf{y}_t \mid D^{t-1}, \gamma, \alpha, \beta, \lambda) \cdot p(\gamma)$, where the product terms are the one-step-ahead predictive densities from the forward filter with $\eta$ integrated out. A Beta proposal distribution centered at the current value provides a natural bounded proposal \citep{GL06}.

\item \textbf{Sampling $\alpha$ and $\beta$}: These are also updated via Metropolis steps using the same marginal likelihood with $\eta$ integrated out.
\end{enumerate}

When $(\alpha, \gamma)$ are fixed---for example, via empirical Bayes selection through the grid search of Section~\ref{sec:application}---steps (iii) and (iv) are unnecessary and the sampler reduces to alternating FFBS for $\eta^T$ and conjugate Gamma draws for $\lambda_j$, both of which are exact (no Metropolis steps). This \emph{conditional Gibbs sampler} mixes rapidly due to the fully conjugate structure and provides posterior distributions for the latent environment process and segment rate parameters that quantify parameter uncertainty beyond point estimates.

\subsection{Particle Filtering for Real-Time Updating}

For real-time updating, the Gibbs sampler is impractical since it requires reprocessing all data each time a new observation arrives. When $(\alpha, \gamma, \lambda)$ are fixed (e.g., selected by grid search as in Section~\ref{sec:application}), the conjugate updating in Proposition~\ref{prop:updating} provides exact online inference for $\eta_t$ with no approximation. When $\Psi$ must also be estimated online, we use particle learning methods \citep{Nick2010} that recursively update the posterior of $(\eta_t, \Psi)$ from time $t-1$ to time $t$.

At each time step, a set of weighted particles $\{(\eta_t^{(i)}, \Psi^{(i)}, w_t^{(i)})\}_{i=1}^N$ represents the posterior distribution. When new data $(y_{1t},\ldots,y_{mt}, u_t)$ arrives:
\begin{enumerate}
\item Propagate each particle's environment: $\eta_t^{(i)} \sim \mathrm{Gam}(\gamma a_{t-1}^{(i)}, \gamma b_{t-1}^{(i)})$.
\item Update weights using the observation likelihood.
\item Resample if the effective sample size drops below a threshold.
\item Update sufficient statistics for $\lambda_j$: since $(\lambda_j \mid \eta^t, D^t) \sim \mathrm{Gam}(r_{jt}, q_{jt})$ with $r_{jt} = r_{j,t-1} + \alpha$ and $q_{jt} = q_{j,t-1} + \eta_t y_{jt} e^{\beta^T u_t}$, the update requires only incrementing the sufficient statistics.
\end{enumerate}

The particle filter provides real-time posterior distributions and predictive densities, making it suitable for operational deployment in traffic management systems \citep{PolsonSokolov2015,Mihaylova2007}.

\section{Application: I-55 Stevenson Expressway}\label{sec:application}

\subsection{Data Description}

We apply the model to traffic data from the Stevenson Expressway (I-55) northbound in Chicago, Illinois. The data come from the Gary-Chicago-Milwaukee (GCM) corridor traffic management system and consist of hourly speed, volume, and occupancy readings from loop detector stations.

\begin{figure}[htbp]
\centering
\includegraphics[width=0.7\textwidth]{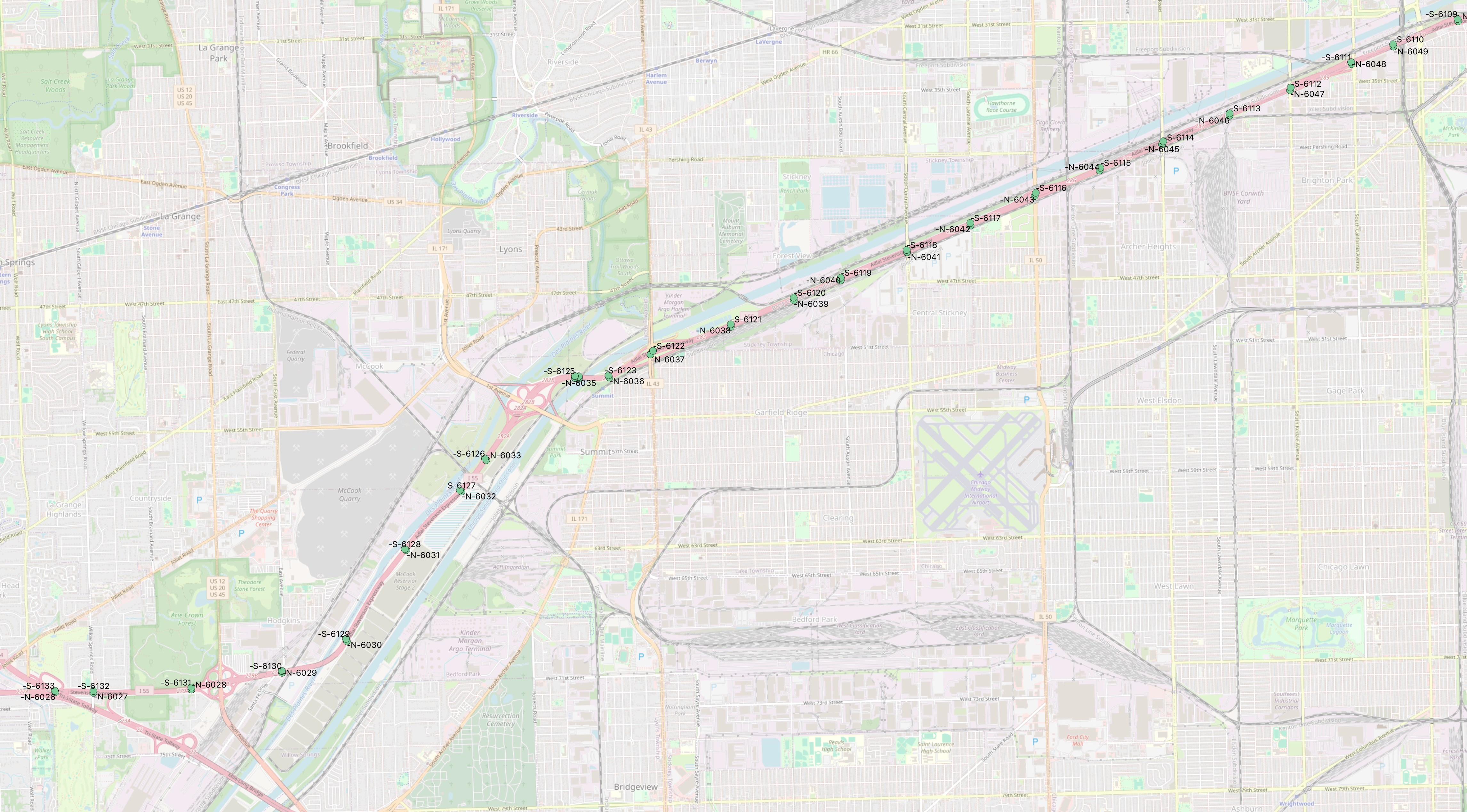}
\caption{Locations of 16 loop detector stations on I-55 northbound (Stevenson Expressway), spanning 8.26 miles from sensor 6030 near I-294 to sensor 6045 near Damen Avenue.}
\label{fig:map}
\end{figure}

Our study corridor consists of 16 consecutive loop detector stations (sensors 6030--6045) spanning 8.26 miles of the northbound Stevenson Expressway (Figure~\ref{fig:map}). Sensors 6034 and 6035 are co-located (31~m apart), resulting in one near-zero-length segment; they are both retained because each reports independent speed measurements. Table~\ref{tab:segments} summarizes the segment distances computed from sensor coordinates via the Haversine formula and the mean travel times.

\begin{table}[htbp]
\centering
\caption{Segment distances and mean travel times (Wednesday 2--8pm, 2019). The $\hat\lambda$ values are rounded; $\alpha^* = 13.3$ in Table~\ref{tab:multivariate} is computed from full-precision values.}
\label{tab:segments}
\begin{tabular}{lrrrrr}
\toprule
Sensor & Distance (mi) & Mean speed (mph) & Mean TT (min) & Gamma $\hat\alpha$ & $\hat\lambda$ \\
\midrule
6030 & 0.853 & 28.5 & 1.84 & 34.2 & 0.46 \\
6031 & 0.632 & 17.1 & 2.44 & 9.3 & 0.35 \\
6032 & 0.313 & 21.2 & 0.96 & 11.1 & 0.88 \\
6033 & 0.949 & 19.8 & 3.14 & 10.7 & 0.27 \\
6034 & 0.031 & 20.4 & 0.10 & 7.9 & 8.34 \\
6035 & 0.239 & 16.6 & 1.04 & 4.9 & 0.82 \\
6036 & 0.370 & 20.3 & 1.20 & 11.0 & 0.71 \\
6037 & 0.666 & 14.7 & 3.49 & 4.1 & 0.24 \\
6038 & 0.536 & 15.0 & 2.47 & 7.5 & 0.34 \\
6039 & 0.400 & 22.0 & 1.16 & 16.4 & 0.73 \\
6040 & 0.565 & 19.2 & 1.83 & 26.1 & 0.47 \\
6041 & 0.539 & 19.4 & 1.93 & 6.3 & 0.44 \\
6042 & 0.555 & 20.9 & 1.73 & 11.2 & 0.49 \\
6043 & 0.547 & 20.5 & 1.92 & 5.3 & 0.44 \\
6044 & 0.530 & 16.6 & 2.09 & 12.6 & 0.41 \\
6045 & 0.530 & 23.3 & 1.46 & 14.7 & 0.59 \\
\midrule
Total & 8.26 &  & 28.80 & & \\
\bottomrule
\end{tabular}
\end{table}

We focus on Wednesday afternoons during 2019 (the year with the most complete data), using seven one-hour periods per day (2:00--2:59pm through 8:00--8:59pm). This gives $52 \times 7 = 364$ potential time periods; after removing periods with missing data from any sensor, 248 complete observations remain.

Segment-level travel times are computed as $y_{jt} = d_j / v_{jt}$, where $d_j$ is the distance from sensor $j$ to the next downstream sensor (in miles) and $v_{jt}$ is the hourly average speed at sensor $j$ (in mph), giving travel time in hours which we convert to minutes. For the last sensor (6045) the distance is set equal to that of the preceding segment. The route travel time is $S_t = \sum_{j=1}^{m} y_{jt}$ with $m = 16$.

\subsection{Static Distribution Analysis}

\begin{figure}[htbp]
\centering
\includegraphics[width=\textwidth]{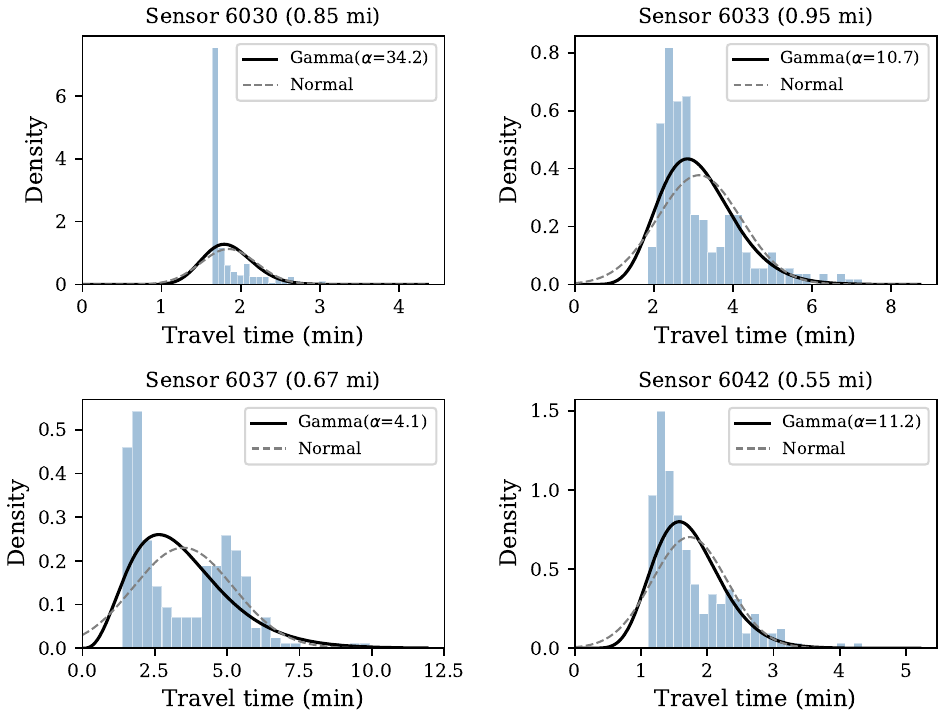}
\caption{Segment-level travel time distributions for four representative sensors, ranging from low variability (sensor 6030, $\hat\alpha = 34.2$) to high variability (sensor 6037, $\hat\alpha = 4.1$). The Gamma distribution (solid) captures the right-skewness present in the data, while the Normal (dashed) cannot accommodate the positive support and asymmetry. Each segment exhibits distinct distributional characteristics, motivating the segment-specific $\lambda_j$ parameters.}
\label{fig:segment_fits}
\end{figure}

Figure~\ref{fig:segment_fits} illustrates the segment-level distributional heterogeneity across four representative sensors. Each sensor has a distinct shape, ranging from the nearly symmetric profile of sensor 6030 ($\hat\alpha = 34.2$, short segment with moderate speeds) to the heavily right-skewed sensor 6037 ($\hat\alpha = 4.1$, a recurring bottleneck area). The Gamma distribution provides a good fit across this full range of shapes, while the Normal distribution fails to capture the positive support and skewness.

\begin{figure}[htbp]
\centering
\includegraphics[width=\textwidth]{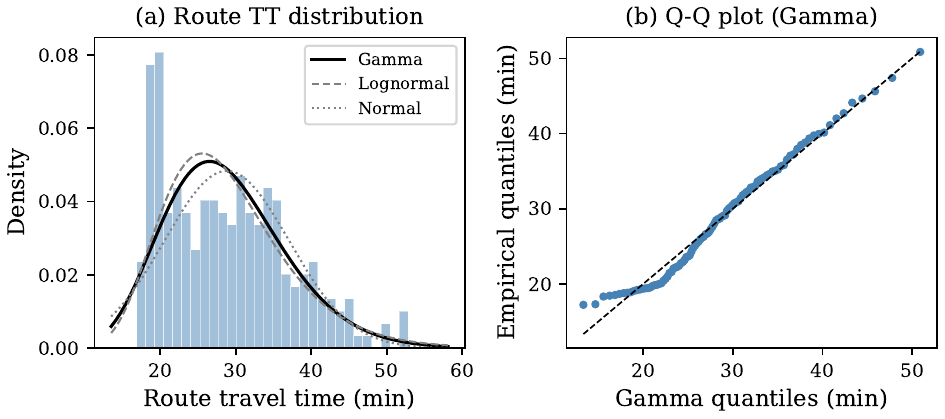}
\caption{(a) Route-level travel time distribution with Gamma, Lognormal, and Normal fits. (b) Q-Q plot of empirical quantiles against the fitted Gamma distribution.}
\label{fig:route_dist}
\end{figure}

The route travel time has mean 28.8 minutes, standard deviation 8.2 minutes, and ranges from 16.9 to 52.8 minutes (Figure~\ref{fig:route_dist}). Table~\ref{tab:static_fit} shows the goodness-of-fit for several static distributions. The Gamma has the highest KS $p$-value, while the inverse Gaussian achieves the lowest AIC and BIC; all three skewed distributions substantially outperform the symmetric normal, which fails due to its support on $(-\infty, \infty)$, consistent with the findings of \citet{Susilawati2013} and \citet{Zang2022}.

\begin{figure}[htbp]
\centering
\includegraphics[width=0.8\textwidth]{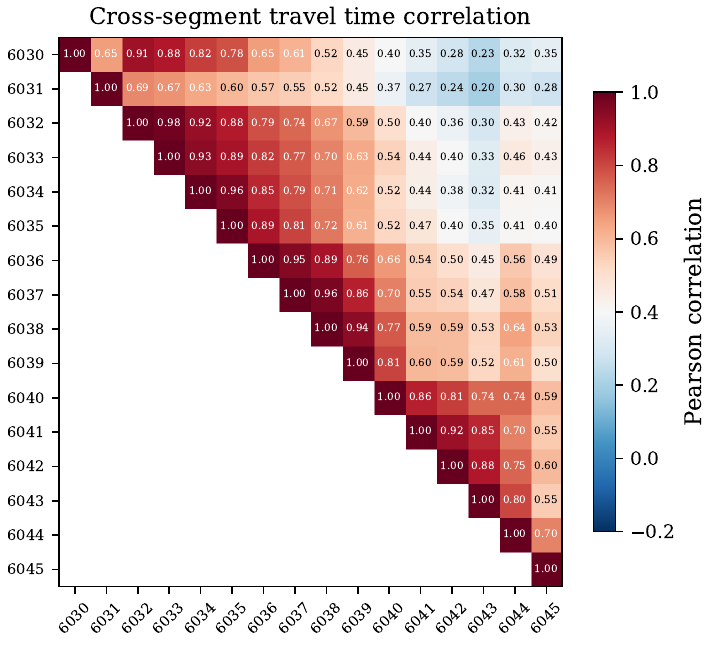}
\caption{Cross-segment travel time correlation matrix. Adjacent segments show strong positive correlation (0.6--0.8), and even distant segments remain positively correlated, supporting the common random environment assumption.}
\label{fig:correlation}
\end{figure}

An important empirical observation is the strong positive correlation among segment travel times (Figure~\ref{fig:correlation}). Adjacent segments exhibit correlations of 0.6--0.8, and even distant segments show correlations of 0.2--0.4. This pervasive positive dependence is precisely what the common random environment $\eta_t$ is designed to capture: shared traffic conditions (congestion waves, weather, demand shocks) that affect the entire corridor simultaneously.

\begin{table}[htbp]
\centering
\caption{Static distribution fits to route travel time (248 observations).}
\label{tab:static_fit}
\begin{tabular}{lrrrr}
\toprule
Distribution & KS statistic & KS $p$-value & AIC & BIC \\
\midrule
Gamma & 0.080 & 0.077 & 1732.0 & 1739.0 \\
Lognormal & 0.085 & 0.055 & 1727.7 & 1734.7 \\
Inv.\ Gaussian & 0.084 & 0.055 & 1726.0 & 1733.0 \\
Weibull & 0.092 & 0.029 & 1758.3 & 1765.4 \\
\bottomrule
\end{tabular}
\end{table}

\paragraph{Bimodality and regime structure.} While the single-component Gamma fits the route distribution adequately (Table~\ref{tab:static_fit}), it masks an important feature. Figure~\ref{fig:mixture_decomposition} shows Gamma mixture fits with $K = 1, 2, 3$ components. A two-component mixture ($K = 2$, BIC $= 1698$ versus $1739$ for $K = 1$) reveals a clear bimodal structure: a narrow free-flow component centered at $\mu = 19$ minutes (weight $0.16$) and a broader congested component centered at $\mu = 31$ minutes (weight $0.84$). The third component ($K = 3$, BIC $= 1710$) does not improve the fit. This bimodality reflects the alternation between distinct traffic regimes---free-flow and congested states---across the observation period. A single static distribution averages over these regimes; the dynamic model developed in Section~\ref{sec:model} addresses this directly, since the time-varying environment $\eta_t$ shifts the predictive distribution toward the free-flow mode when conditions are good and toward the congested mode when conditions deteriorate.

\begin{figure}[htbp]
\centering
\includegraphics[width=\textwidth]{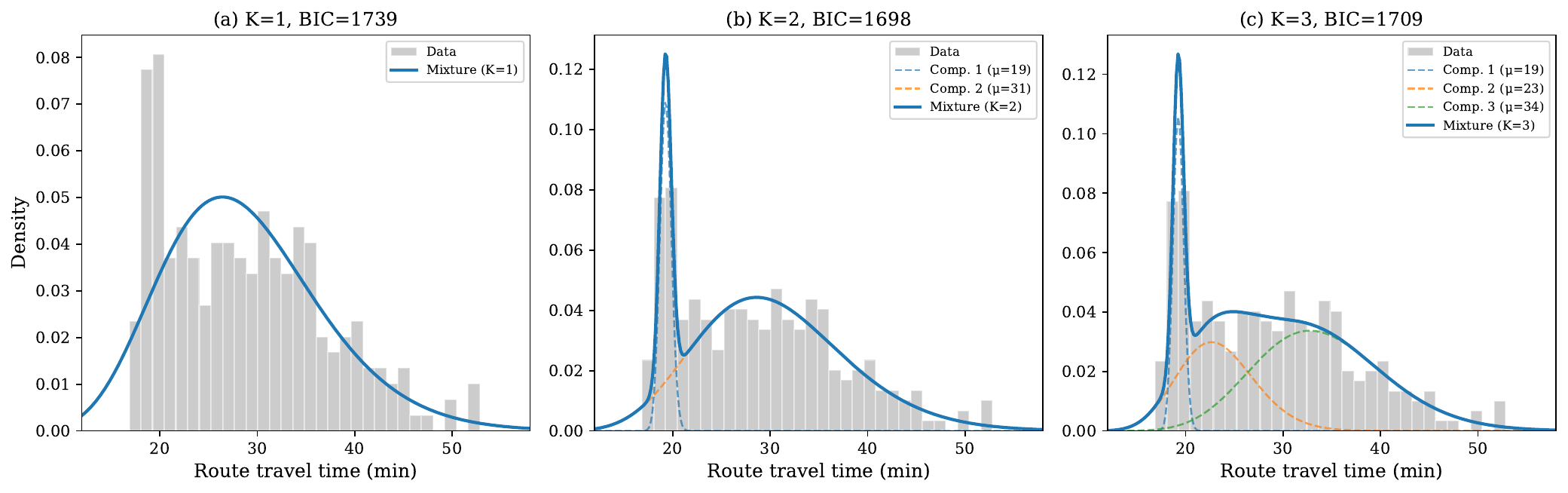}
\caption{Gamma mixture decomposition of route travel time. (a)~Single Gamma ($K=1$). (b)~Two-component mixture ($K=2$, BIC-preferred): a narrow free-flow component ($\mu = 19$ min) and a broad congested component ($\mu = 31$ min). (c)~Three components ($K=3$) do not improve BIC. The bimodality reflects regime switching between free-flow and congested states.}
\label{fig:mixture_decomposition}
\end{figure}

Per-segment Gamma fits reveal substantial heterogeneity in the shape parameter $\hat\alpha$, ranging from 4.1 (sensor 6037, high variability) to 34.2 (sensor 6030, low variability). The segment-specific $\lambda_j$ parameters in our multivariate model capture this heterogeneity through the segment means, while the common environment $\eta_t$ links all segments through shared conditions. We note, however, that the multivariate model enforces a \emph{common} shape parameter $\alpha$ across all segments; only the rate $\lambda_j \eta_t$ varies by segment. This is the cost of preserving the conjugate structure that yields the closed-form route distribution. The best-fitting common $\alpha = 1$ (Section~\ref{sec:application}) is far from the per-segment MLEs, so the segment-level Gamma shapes are only approximately captured. Despite this trade-off, the route-level calibration remains excellent because the moment-matching aggregation over $m$ segments produces an effective route shape $\alpha^*$ (equation~\ref{eq:alpha_star}) that is well-suited for route-level prediction, even when the common $\alpha$ does not perfectly match each individual segment.

\subsection{Empirical Bayes Selection of $(\alpha, \gamma)$}\label{sec:hyperpar}

In this application, we set the covariate vector $u_t = 0$ (equivalently, $\beta = 0$), so that all time-varying conditions are absorbed by the latent environment $\eta_t$. While incorporating observed covariates (weather, incidents) via the $e^{\beta^T u_t}$ term is straightforward in principle, the hourly loop-detector data do not include concurrent covariate measurements.

Our inference proceeds in two stages. First, we select the hyperparameters $(\alpha, \gamma)$ via empirical Bayes, evaluating route-level predictive calibration across a grid. Second, conditional on the selected $(\alpha, \gamma)$, we run the conditional Gibbs sampler (Section~\ref{sec:inference}) to obtain full posterior distributions for the segment rate parameters $\lambda_j$ and the latent environment process $\eta_t$. The two-stage approach is motivated by the fact that the full joint posterior of $(\alpha, \gamma, \lambda, \eta^T)$ concentrates at a mode that optimizes the segment-level marginal likelihood but produces poor route-level predictions; the empirical Bayes step selects $(\alpha, \gamma)$ to optimize route-level calibration instead.

For the grid search, we initialize $\lambda_j$ proportional to $1/\bar{y}_j$ (the inverse of the segment sample mean), normalized to have unit mean, resolving the scale ambiguity discussed in Section~\ref{sec:multivariate}. We evaluate two approaches: a univariate model applied directly to route travel time, and the multivariate model applied to individual segment travel times with route reliability computed via the moment-matching formula~\eqref{eq:route_F}.

For model assessment, we use three criteria:
\begin{enumerate}
\item \textbf{Log predictive likelihood:} $\sum_{t > t_0} \log p(y_t \mid D^{t-1})$, the sum of one-step-ahead log predictive densities after a burn-in of $t_0 = 30$.
\item \textbf{PIT calibration:} The Probability Integral Transform values $u_t = P(Y_t \leq y_t \mid D^{t-1})$ should be Uniform$(0,1)$ if the model is well-calibrated. We use the Kolmogorov--Smirnov test of uniformity.
\item \textbf{Coverage:} The fraction of observations falling within the 90\% predictive interval.
\end{enumerate}

\paragraph{Univariate route model.} Table~\ref{tab:univariate} shows results from a grid search over $\alpha$ and $\gamma$. The model with $\alpha = 10$ provides well-calibrated predictions across a range of discount factors, with 90\% interval coverage consistently near the nominal level.

\begin{table}[htbp]
\centering
\caption{Univariate dynamic Gamma model for route travel time: selected parameter combinations.}
\label{tab:univariate}
\begin{tabular}{rrrrrr}
\toprule
$\alpha$ & $\gamma$ & Log pred.\ lik. & PIT KS stat & KS $p$-value & 90\% coverage \\
\midrule
10 & 0.80 & $-774.2$ & 0.082 & 0.103 & 0.977 \\
10 & 0.85 & $-771.7$ & 0.071 & 0.211 & 0.977 \\
10 & 0.90 & $-769.7$ & 0.075 & 0.165 & 0.972 \\
10 & 0.95 & $-768.0$ & 0.074 & 0.171 & 0.968 \\
10 & 0.99 & $-768.2$ & 0.068 & 0.247 & 0.963 \\
\bottomrule
\end{tabular}
\end{table}

\paragraph{Multivariate model with route reliability.} The multivariate model is applied to the $m$-dimensional vector of segment travel times, and route reliability is computed via the closed-form $F$-distribution formula~\eqref{eq:route_F}. Figure~\ref{fig:heatmap} shows the route-level PIT calibration (KS $p$-value) across the $(\alpha, \gamma)$ grid. A clear region of good calibration emerges at $\alpha \in [1.0, 1.5]$ and $\gamma \in [0.50, 0.70]$, with the best calibration at $\alpha = 1.0$, $\gamma = 0.70$ (KS $p = 0.467$). Table~\ref{tab:multivariate} shows detailed results for the best-calibrated combinations. These values are used in the subsequent Bayesian posterior analysis (Section~\ref{sec:posterior}).

\begin{figure}[htbp]
\centering
\includegraphics[width=0.8\textwidth]{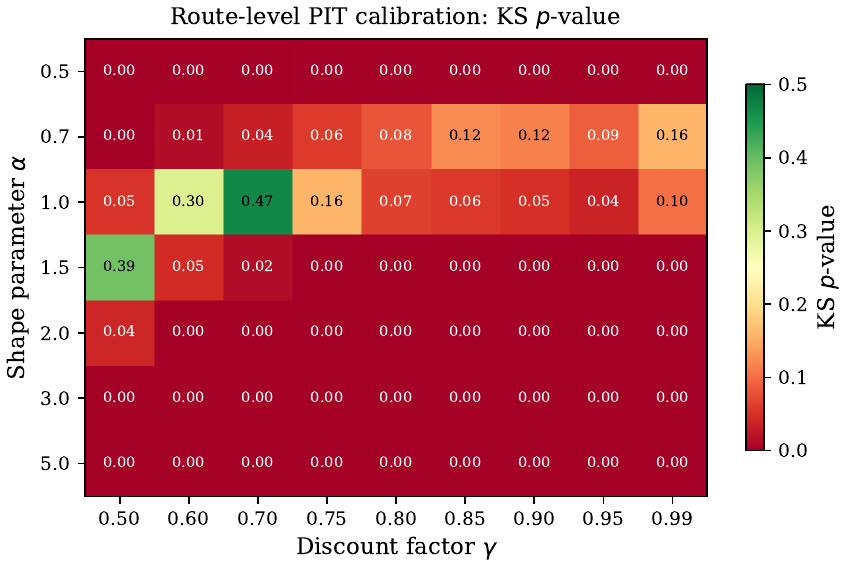}
\caption{Empirical Bayes selection of $(\alpha, \gamma)$: route-level PIT calibration (KS $p$-value) across the hyperparameter grid. Green cells indicate well-calibrated models ($p > 0.05$). The selected values $\alpha = 1.0$, $\gamma = 0.70$ ($p = 0.467$) are used for subsequent Bayesian posterior inference.}
\label{fig:heatmap}
\end{figure}

\begin{table}[htbp]
\centering
\caption{Multivariate dynamic Gamma model: route-level calibration via moment matching.}
\label{tab:multivariate}
\begin{tabular}{rrrrrrr}
\toprule
$\alpha$ & $\gamma$ & Joint log lik. & Route KS & KS $p$ & Route 90\% cov. & $\alpha^*$ \\
\midrule
1.00 & 0.70 & $-4927$ & 0.057 & 0.467 & 0.954 & 13.3 \\
1.50 & 0.50 & $-4203$ & 0.060 & 0.389 & 0.931 & 19.9 \\
1.00 & 0.60 & $-4930$ & 0.065 & 0.295 & 0.972 & 13.3 \\
1.00 & 0.65 & $-4928$ & 0.067 & 0.269 & 0.963 & 13.3 \\
0.70 & 0.99 & $-5603$ & 0.075 & 0.159 & 0.963 & 9.3 \\
\bottomrule
\end{tabular}
\end{table}

The multivariate model with $\alpha = 1.0$ (corresponding to exponential segment travel times given $\eta_t$) and $\gamma = 0.70$ achieves the best route calibration: the PIT KS test yields a nominal $p$-value of 0.47 (though see the autocorrelation caveat in Section~\ref{sec:application}), and the 90\% interval coverage is 95.4\%. The effective route shape parameter $\alpha^* = 13.3$ (using MLE $\lambda_j$) provides a well-shaped predictive distribution. We adopt these values for the full Bayesian posterior analysis in Section~\ref{sec:posterior}.

The runner-up $\alpha = 1.5$, $\gamma = 0.50$ also achieves good calibration (KS $p = 0.389$) with a faster-adapting environment process. The choice of $\gamma = 0.50$--$0.75$ implies that the model gives significant weight to recent observations: with hourly data, $\gamma = 0.70$ corresponds to an effective memory of roughly $1/(1-\gamma) \approx 3.3$ hours, meaning that conditions from earlier in the afternoon are substantially discounted by evening. This short memory is consistent with the rapidly changing conditions on an urban expressway during peak hours, where congestion can build and dissipate within 1--2 hours.

An important finding is that while the best univariate model has $\alpha = 10$ (moderately peaked Gamma shape for route TT), the best multivariate model has $\alpha = 1$--$1.5$ at the segment level. The route-level shape $\alpha^* = 13$--$20$ emerges from aggregating across $m = 16$ segments, demonstrating how the moment-matching formula~\eqref{eq:alpha_star} appropriately scales up the effective shape parameter. We note that $\alpha = 1$ corresponds to exponential segment travel times conditional on $\eta_t$---the most variable member of the Gamma family. The Gamma moment-matching approximation for the sum of $m$ independent exponentials is known to be accurate even for moderate $m$ \citep{MoschouSingle}, and with $m = 16$ the resulting $\alpha^* \approx 13$ produces a well-shaped route distribution.

\subsection{Bayesian Posterior Analysis}\label{sec:posterior}

With $(\alpha, \gamma) = (1.0, 0.70)$ selected via empirical Bayes (Section~\ref{sec:hyperpar}), we run the conditional Gibbs sampler of Section~\ref{sec:inference} to obtain full posterior distributions for the segment rate parameters $\lambda_j$ and the latent environment process $\eta^T = (\eta_1, \ldots, \eta_T)$. Each iteration consists of two conjugate steps: (i)~forward filtering and FFBS backward sampling \citep{FS94,APS20} of $\eta^T$, and (ii)~drawing each $\lambda_j$ from its Gamma full conditional given $\eta^T$. Four chains of 10{,}000 iterations (2{,}000 burn-in, thinning by 2) are run from different initial values, yielding 16{,}000 posterior samples. The fully conjugate structure produces excellent mixing: the maximum split-$\hat{R}$ across all $\lambda_j$ is $1.0004$ and the minimum effective sample size is $3{,}722$.

Table~\ref{tab:lambda_posterior} reports the posterior summaries for $\lambda_j$. The posterior means are close to the MLE values from Table~\ref{tab:segments}, with narrow 95\% credible intervals that confirm the data are informative for these parameters. The posterior for the route-level shape parameter is $\alpha^* = 13.26$ with 95\% CrI $[12.91, 13.57]$, consistent with the plug-in value of 13.3.

\begin{table}[htbp]
\centering
\caption{Posterior summaries for segment rate parameters $\lambda_j$ from the conditional Gibbs sampler ($\alpha = 1.0$, $\gamma = 0.70$). Posterior means, 95\% credible intervals, and MLE values.}
\label{tab:lambda_posterior}
\begin{tabular}{lrrr}
\toprule
Sensor & Post.\ mean & 95\% CrI & MLE \\
\midrule
6030 & 0.456 & $[0.395, 0.520]$ & 0.462 \\
6031 & 0.349 & $[0.302, 0.399]$ & 0.350 \\
6032 & 0.882 & $[0.769, 1.005]$ & 0.884 \\
6033 & 0.272 & $[0.235, 0.311]$ & 0.272 \\
6034 & 8.359 & $[7.830, 8.864]$ & 8.338 \\
6035 & 0.828 & $[0.722, 0.944]$ & 0.820 \\
6036 & 0.711 & $[0.618, 0.812]$ & 0.713 \\
6037 & 0.247 & $[0.213, 0.284]$ & 0.245 \\
6038 & 0.345 & $[0.300, 0.396]$ & 0.345 \\
6039 & 0.728 & $[0.632, 0.831]$ & 0.734 \\
6040 & 0.461 & $[0.401, 0.528]$ & 0.466 \\
6041 & 0.444 & $[0.383, 0.507]$ & 0.443 \\
6042 & 0.490 & $[0.424, 0.560]$ & 0.492 \\
6043 & 0.442 & $[0.383, 0.504]$ & 0.444 \\
6044 & 0.405 & $[0.352, 0.464]$ & 0.409 \\
6045 & 0.581 & $[0.504, 0.664]$ & 0.586 \\
\bottomrule
\end{tabular}
\end{table}

\begin{figure}[htbp]
\centering
\includegraphics[width=\textwidth]{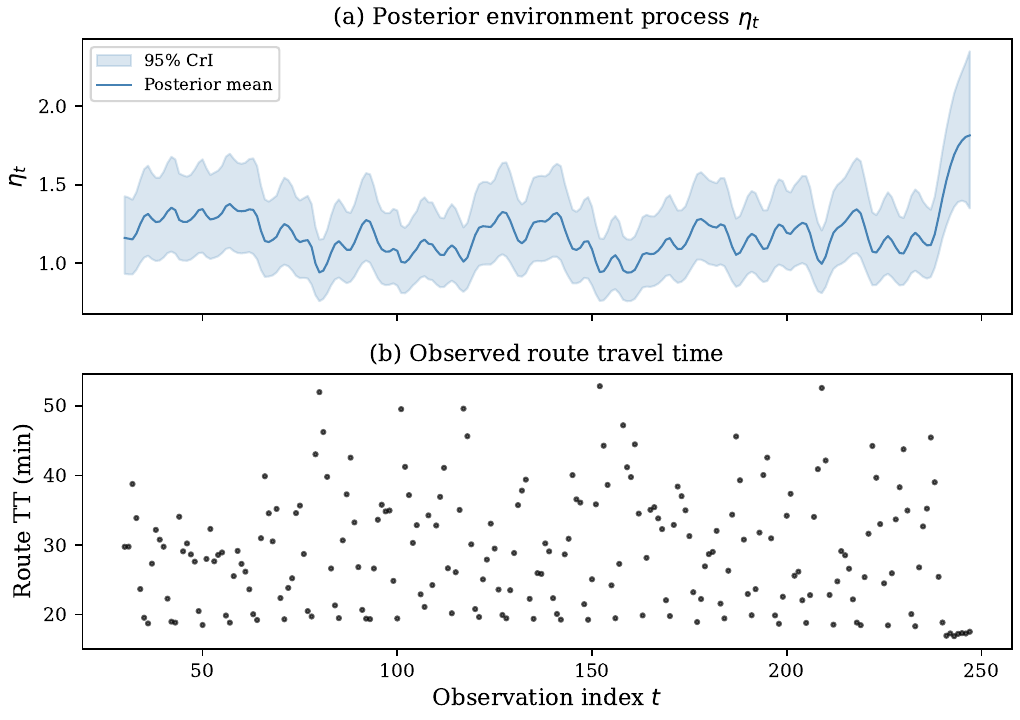}
\caption{(a) Posterior mean and 95\% credible interval of the common environment process $\eta_t$ from the conditional Gibbs sampler. (b) Observed route travel time for reference. The environment process is inversely related to congestion: $\eta_t$ rises during free-flow and dips during congested periods.}
\label{fig:eta_posterior}
\end{figure}

\begin{figure}[htbp]
\centering
\includegraphics[width=0.7\textwidth]{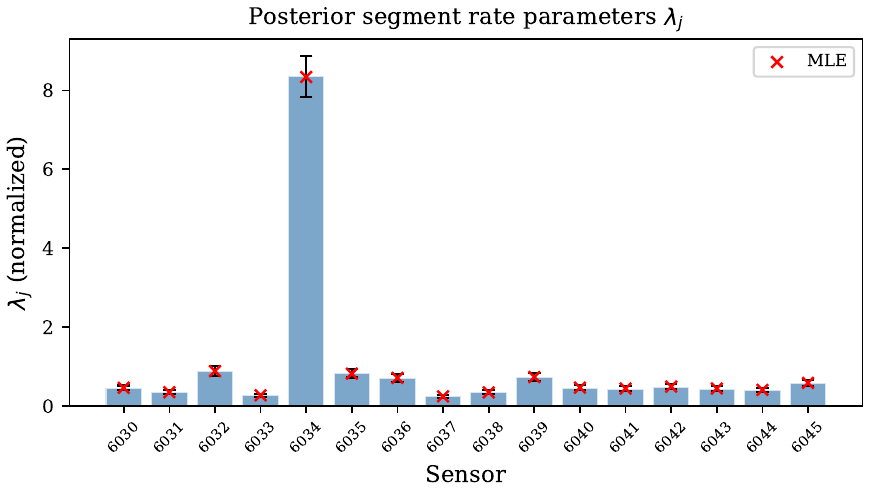}
\caption{Posterior distributions of segment rate parameters $\lambda_j$: bars show posterior means with 95\% credible intervals (error bars), and crosses show the MLE values. The posterior concentrates tightly around the MLE for all segments.}
\label{fig:lambda_posterior}
\end{figure}

Figure~\ref{fig:eta_posterior} shows the posterior trajectory of $\eta_t$ with pointwise 95\% credible intervals. The environment process tracks inversely with observed congestion, rising during free-flow (low travel times) and declining during congested periods. The credible intervals are narrow relative to the posterior mean, indicating that the filtering posterior pins down the environment state at each time step with limited uncertainty.

Figure~\ref{fig:lambda_posterior} shows the posterior distributions of $\lambda_j$ alongside the MLE values. The posterior concentrates tightly around the MLE for all 16 segments, with the largest relative uncertainty for sensor 6037 (the segment with the highest mean travel time and lowest $\lambda_j$). The large $\lambda_{6034} \approx 8.4$ reflects the near-zero distance of the co-located sensor pair 6034/6035.

The Bayesian predictive distribution, averaged over posterior samples of $(\lambda, \eta^T)$, yields a route-level KS $p$-value of 0.41 and 90\% coverage of 95.4\%---virtually identical to the plug-in results. This confirms that the data are sufficiently informative to make the plug-in and Bayesian predictive distributions effectively equivalent for this application. The practical implication is that the closed-form conjugate updating of Section~\ref{sec:route} can be used directly for real-time reliability computation with negligible loss relative to full Bayesian inference.

\subsection{Route Reliability Results}

Using the multivariate model with $(\alpha, \gamma) = (1.0, 0.70)$ and MLE $\lambda_j$ (which are virtually identical to the posterior means; see Table~\ref{tab:lambda_posterior}), we compute route reliability metrics at each time step via the conjugate online updating of Proposition~\ref{prop:updating}, which is the operationally relevant mode. The free-flow route travel time (5th percentile) is approximately 18 minutes for the 8.26-mile corridor, corresponding to an average speed of 27.5 mph.

\begin{figure}[htbp]
\centering
\includegraphics[width=\textwidth]{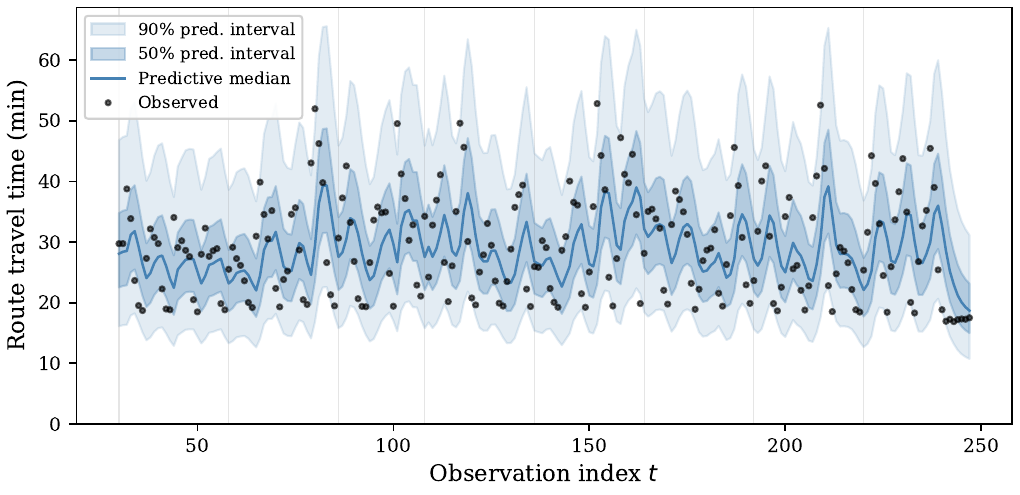}
\caption{Observed route travel times (dots) and one-step-ahead predictive intervals from the multivariate dynamic Gamma model ($\alpha = 1.0$, $\gamma = 0.70$). The 50\% (dark) and 90\% (light) predictive intervals track the evolving conditions, widening during congested periods and narrowing during free-flow.}
\label{fig:route_timeseries}
\end{figure}

Figure~\ref{fig:route_timeseries} shows the observed route travel times overlaid with the one-step-ahead predictive intervals from the multivariate model. The predictive bands adapt to changing conditions: they widen during high-congestion periods (reflecting greater uncertainty) and narrow during free-flow periods. Nearly all observations fall within the 90\% interval (actual coverage: 95.4\%), and the median tracks the central tendency well.

\begin{figure}[htbp]
\centering
\includegraphics[width=\textwidth]{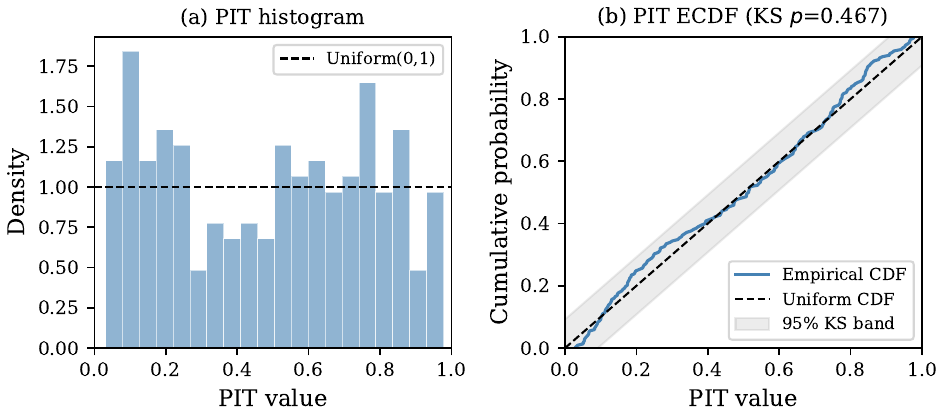}
\caption{Probability integral transform (PIT) calibration for the route-level predictive distribution. (a) PIT histogram: a well-calibrated model produces a uniform histogram. (b) Empirical CDF of PIT values vs.\ the Uniform(0,1) CDF, with the 95\% Kolmogorov--Smirnov confidence band (gray).}
\label{fig:pit}
\end{figure}

Figure~\ref{fig:pit} provides a formal calibration assessment. The PIT histogram is approximately uniform, and the PIT empirical CDF lies entirely within the 95\% KS confidence band, with $p = 0.47$. We note that the 90\% interval achieves 95.4\% coverage---above the nominal 90\%---indicating mildly conservative (over-dispersed) predictive tails. This over-coverage is consistent with the moment-matching approximation, which may slightly overestimate the tail weight of the route-level distribution. The KS test does not reject uniformity at the 5\% level, but with $n = 218$ evaluation points (after burn-in) the test has limited power to detect moderate departures. We also note that the standard KS test assumes i.i.d.\ PIT values, whereas the sequential structure of the predictions induces substantial autocorrelation in the PITs (lag-1 autocorrelation of 0.52; Ljung--Box $Q = 59.8$ on 1 degree of freedom, $p < 10^{-14}$). This autocorrelation arises naturally from the within-day temporal structure: consecutive hours on the same Wednesday afternoon share similar traffic conditions. While the marginal distribution of the PITs appears approximately uniform, the serial dependence means that the KS $p$-value overstates the effective sample size and thus the evidence for calibration. More sophisticated calibration tests that account for temporal dependence \citep[e.g.,][]{Diebold1998} would provide a more rigorous assessment.

\begin{figure}[htbp]
\centering
\includegraphics[width=0.7\textwidth]{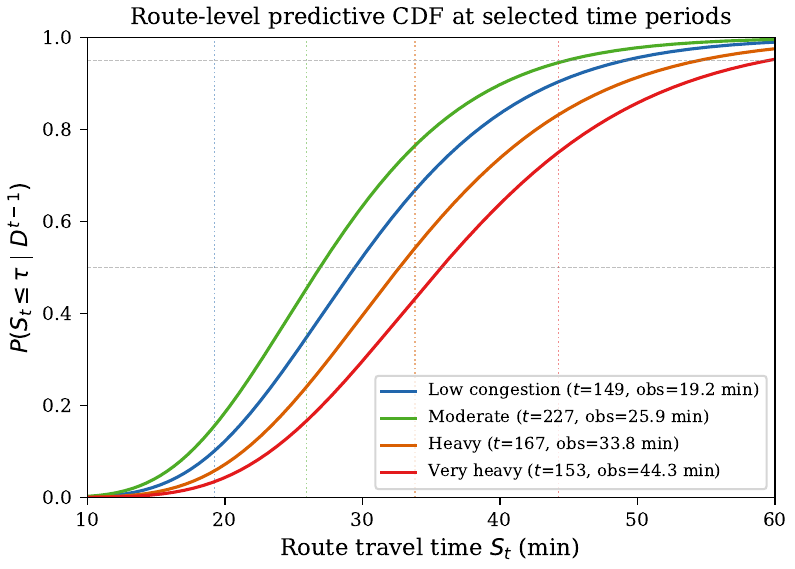}
\caption{Route-level predictive CDF at four selected time periods representing different congestion levels. The closed-form $F$-distribution (Proposition~\ref{prop:route}) shifts right and becomes more dispersed under congestion.}
\label{fig:route_cdf}
\end{figure}

Figure~\ref{fig:route_cdf} illustrates the route predictive CDF at four time periods spanning the congestion spectrum. Under low congestion, the distribution is concentrated around 20 minutes with a steep CDF; under heavy congestion, it shifts to 40+ minutes with a flatter CDF reflecting greater uncertainty. All four CDFs are computed instantly via the $F$-distribution formula~\eqref{eq:route_F}.

\begin{figure}[htbp]
\centering
\includegraphics[width=\textwidth]{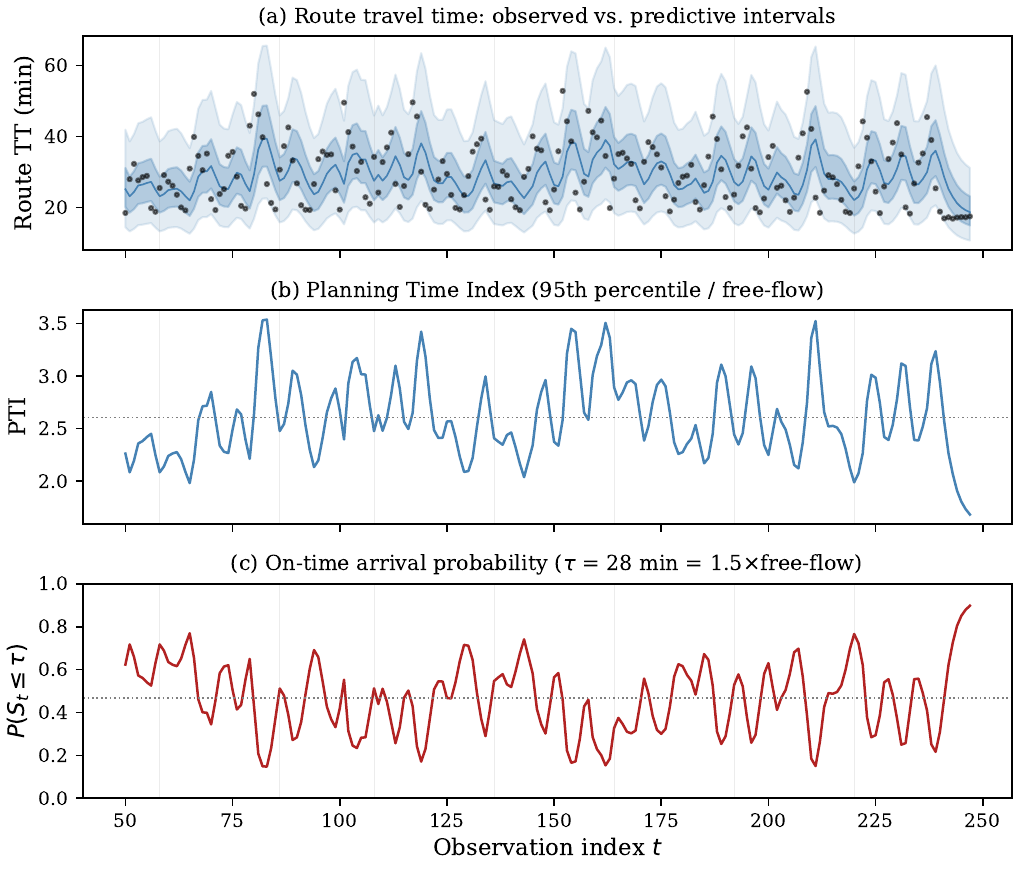}
\caption{Dynamic reliability metrics computed from the multivariate model. (a) Route travel time with predictive intervals. (b) Planning Time Index (95th percentile divided by free-flow travel time). (c) On-time arrival probability for a threshold of 1.5 times free-flow travel time.}
\label{fig:reliability}
\end{figure}

Figure~\ref{fig:reliability} shows three reliability metrics computed in real time from the model's closed-form predictive distribution. The Planning Time Index (PTI) ranges from 2.0 during light traffic to 3.5 during peak congestion, meaning travelers must budget 2 to 3.5 times the free-flow travel time to ensure a 95\% probability of on-time arrival. The on-time arrival probability (panel c) varies substantially across time periods, falling below 30\% during the worst congested periods. These metrics update at each time step with $O(1)$ computation.

\begin{figure}[htbp]
\centering
\includegraphics[width=\textwidth]{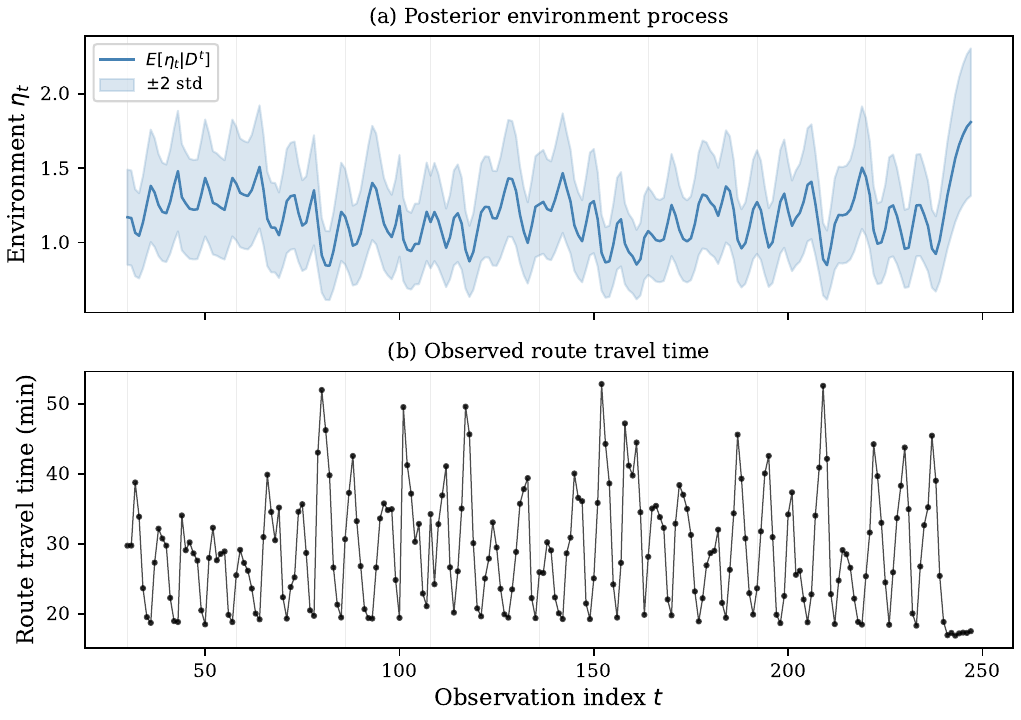}
\caption{(a) Posterior mean and $\pm 2$ standard deviation band of the common environment $\eta_t$. Higher values indicate better conditions (lower travel times). (b) Corresponding observed route travel times. The environment process tracks the observed congestion dynamics.}
\label{fig:environment}
\end{figure}

Figure~\ref{fig:environment} shows the posterior environment process $\eta_t$ alongside the observed route travel times. The environment process tracks the inverse of congestion: it rises during free-flow periods (low travel times) and dips during congestion (high travel times). This single latent process captures the shared conditions that drive all $m$ segment travel times simultaneously, providing a parsimonious representation of the corridor's state.

\subsection{Comparison with Alternative Approaches}

A central claim of our framework is that it solves the correlated link-to-route aggregation problem---historically the main obstacle to closed-form route reliability---while maintaining the same $O(1)$ computational cost as methods that assume independence. To demonstrate this, we compare against four alternative approaches for constructing the route travel time distribution:
\begin{enumerate}
\item \textbf{Independent Gamma convolution} (analogous to \citealt{Lei2014}): fit Gamma distributions to each segment independently, then obtain the route distribution by moment-matching the sum of independent Gammas.
\item \textbf{Independent Normal convolution} (analogous to \citealt{iida99}): compute segment means and variances, sum them under independence.
\item \textbf{Gaussian copula with Gamma marginals} (analogous to \citealt{ChenYu2017}): fit Gamma marginals to each segment, transform to normal space, estimate the correlation matrix, and obtain the route distribution via Monte Carlo simulation (50,000 draws).
\item \textbf{Static Gamma} (direct): fit a single Gamma distribution to the pooled route travel time data.
\end{enumerate}
The comparison with the independence-based methods (1--2) is particularly informative because our model operates at the \emph{same computational scale}: both evaluate the route CDF via a single distributional formula (Gamma for the independence methods, $F$-distribution for ours), with no simulation or numerical integration required. The critical difference is that our model accounts for cross-segment dependence through the common environment, whereas the independence methods do not. The copula approach (3) does model dependence, but requires Monte Carlo simulation that is orders of magnitude more expensive. All four alternatives are also static: they produce a single, time-invariant route distribution, whereas our model produces a time-varying predictive that adapts as new data arrives.

\begin{table}[htbp]
\centering
\caption{Comparison of route-level predictive performance. KS $p$: Kolmogorov--Smirnov test $p$-value for PIT uniformity (higher is better). Coverage: fraction of observations within the 90\% predictive interval (nominal 0.90). IW: mean width of the 90\% interval.}
\label{tab:comparison}
\begin{tabular}{lccccc}
\toprule
Method & Dynamic & Depend. & KS $p$ & 90\% Cov. & IW (min) \\
\midrule
\textbf{Dynamic Gamma (ours)} & \textbf{Yes} & \textbf{Yes} & \textbf{0.467} & \textbf{0.954} & \textbf{31.5} \\
Gaussian copula + Gamma & No & Yes & 0.082 & 0.913 & 25.8 \\
Static Gamma (direct) & No & N/A & 0.162 & 0.931 & 26.4 \\
Indep.\ Gamma convolution & No & No & 0.000 & 0.344 & 9.4 \\
Indep.\ Normal convolution & No & No & 0.000 & 0.367 & 10.1 \\
\bottomrule
\end{tabular}
\end{table}

\begin{figure}[htbp]
\centering
\includegraphics[width=\textwidth]{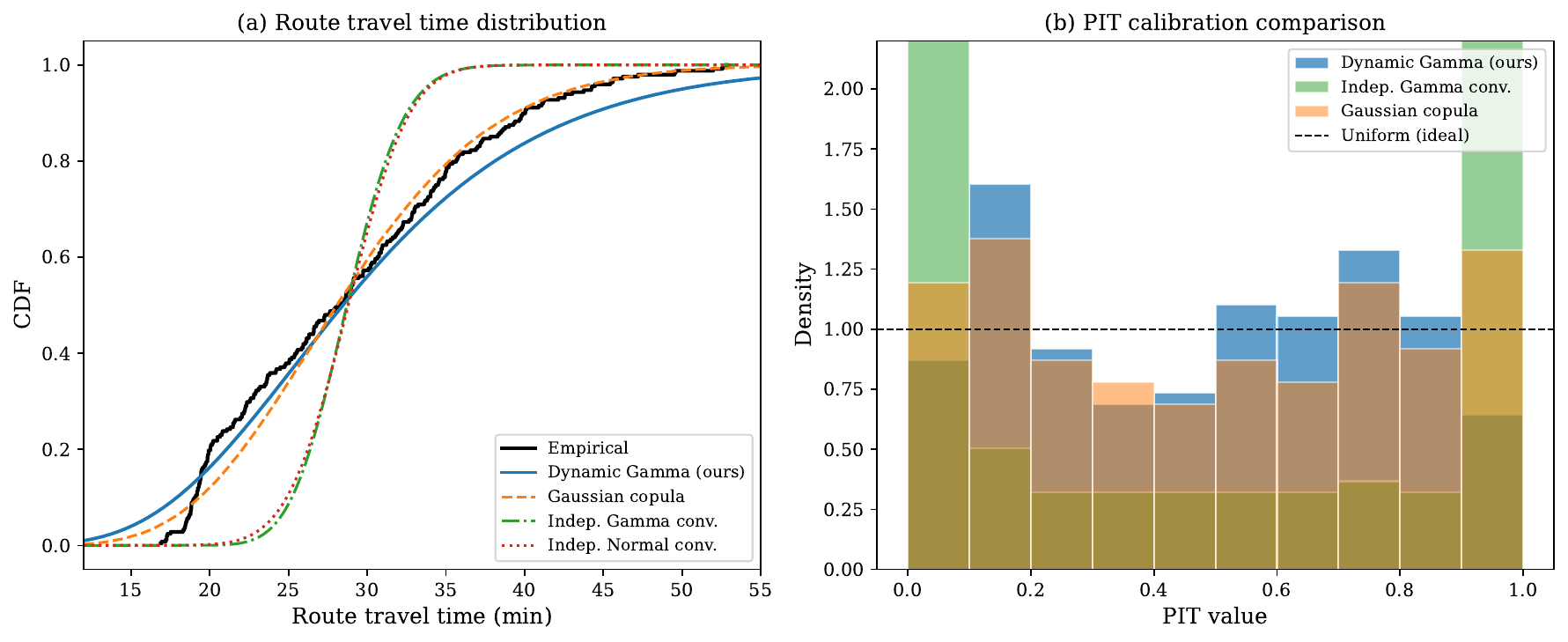}
\caption{(a) Route travel time CDF: the dynamic Gamma model (averaged over time periods) closely tracks the empirical CDF, while independence-based methods produce distributions that are too narrow. (b) PIT calibration comparison: only the dynamic Gamma model produces approximately uniform PIT values.}
\label{fig:comparison_cdf}
\end{figure}

\begin{figure}[htbp]
\centering
\includegraphics[width=\textwidth]{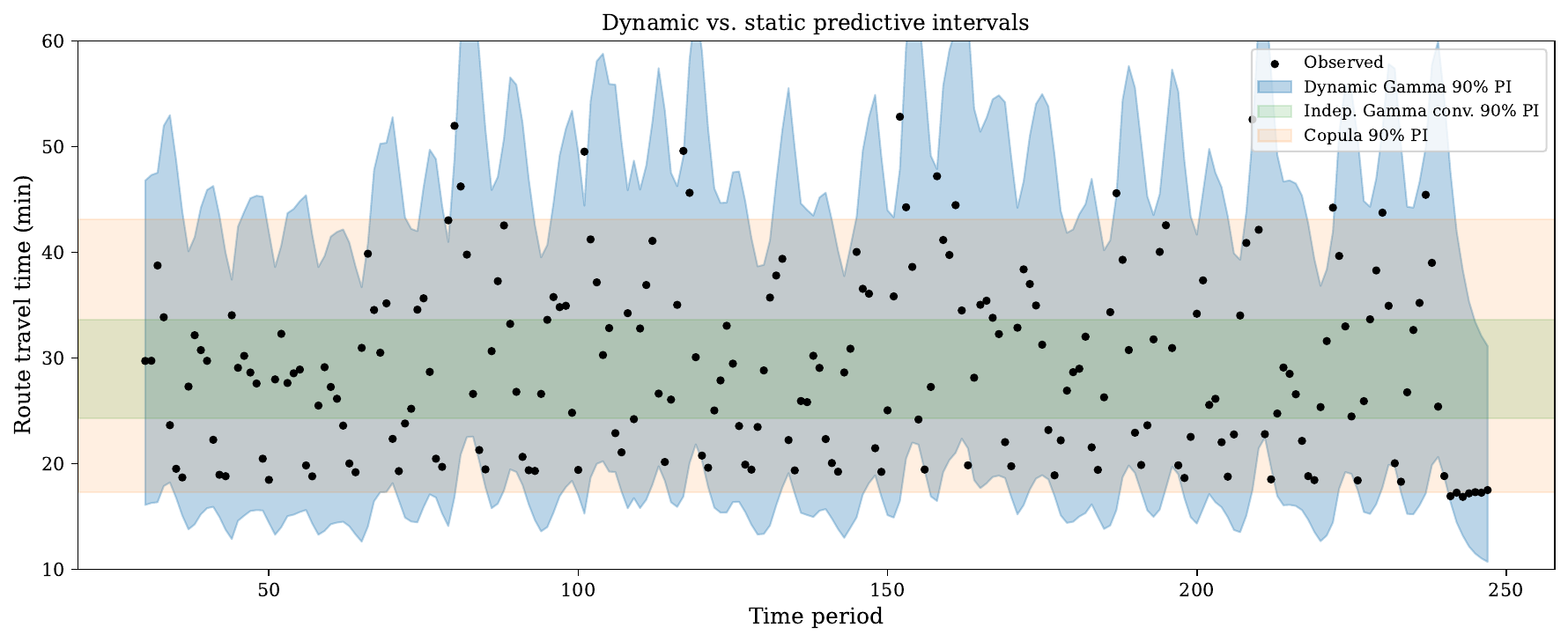}
\caption{Time-varying predictive intervals from the dynamic Gamma model (blue bands) compared with constant-width intervals from the static independence-based Gamma convolution (green) and the Gaussian copula (orange). The dynamic model adapts its uncertainty to current conditions, while static methods either undercover (independence) or cannot respond to changing traffic states (copula).}
\label{fig:comparison_intervals}
\end{figure}

Table~\ref{tab:comparison} reveals a clear ordering. The independence-based methods---Gamma convolution and Normal convolution---fail catastrophically, achieving only 34--37\% coverage of the nominal 90\% predictive interval. Their interval widths of 9--10 minutes are far too narrow because they ignore cross-segment dependence: by summing segment variances under independence, they underestimate the route variance by a factor of roughly $\sum_{j,k}\mathrm{Cov}(y_j,y_k)/\sum_j\mathrm{Var}(y_j) \approx 2.7$. This dramatic under-coverage illustrates why the link-to-route aggregation problem cannot be solved by simply summing independent segment distributions, regardless of how well each segment is modeled individually.

The Gaussian copula with Gamma marginals correctly captures the dependence structure and achieves 91.3\% coverage with a KS $p$-value of 0.082---a substantial improvement over independence. However, it comes at significant computational cost: generating the route distribution requires drawing 50,000 samples from a $m$-dimensional Gaussian copula, transforming each sample through segment-specific inverse CDFs, and computing the route sum---a procedure that takes on the order of 0.4 seconds on a modern desktop (timed in our implementation). In contrast, our model evaluates the route CDF via a single call to the $F$-distribution (equation~\ref{eq:route_F}), at the same cost as the independence-based methods. The static Gamma fitted directly to route data performs reasonably (93.1\% coverage), but uses no segment-level information and cannot be composed for alternative routes.

Our dynamic Gamma model achieves the best marginal calibration (nominal KS $p = 0.467$; see PIT autocorrelation caveat in Section~\ref{sec:application}) with 95.4\% coverage. Crucially, it obtains this performance while maintaining closed-form computation: the common random environment structure absorbs the cross-segment dependence into the conditional independence framework, so that modeling dependence adds no computational overhead relative to the independence assumption. The time-varying intervals widen during congested periods and narrow during free-flow (Figure~\ref{fig:comparison_intervals}), a behavior that no static method can replicate regardless of how it handles dependence.

We acknowledge that this comparison conflates two advantages of our model: (i) accounting for cross-segment dependence, and (ii) producing time-varying predictions. Disentangling these effects would require a dynamic independence baseline (e.g., fitting independent dynamic Gamma models per segment and convolving under independence). We note, however, that the static Gamma fitted directly to route data---which implicitly captures all dependence in the pooled marginal---achieves only 93.1\% coverage, whereas our model achieves 95.4\%. This suggests that the time-varying component contributes meaningfully beyond dependence modeling alone. We also note that the evaluation design is asymmetric: the four static baselines are fitted on the \emph{full} 248-observation sample and evaluated in-sample, whereas the dynamic Gamma model produces genuine one-step-ahead predictives using only data up to $t-1$. This asymmetry actually \emph{favors} the static baselines; their inferior performance is therefore even more telling. Separately, the hyperparameters $(\alpha, \gamma)$ are selected via empirical Bayes (grid search) on the same 248 observations used for evaluation; a held-out or cross-validated assessment would provide a more stringent test, though the sequential one-step-ahead evaluation already guards against in-sample overfitting.

\section{Discussion and Conclusion}\label{sec:conclusion}

We have presented a dynamic Gamma model with common random environment that provides a principled and computationally efficient solution to the route-level travel time reliability problem. The main theoretical contribution is showing that the common environment structure---which is physically motivated by the shared traffic conditions affecting all segments of a route---resolves the correlated link-to-route aggregation problem without any computational penalty: route reliability metrics that would otherwise require $m$-dimensional integration or Monte Carlo simulation can be computed in closed form via the $F$-distribution, at the same $O(1)$ cost as methods that simply assume independence. The empirical comparison in Section~\ref{sec:application} confirms that this matters: independence-based methods that share the same computational cost achieve only 34--37\% coverage, whereas our model achieves 95.4\% coverage by correctly accounting for cross-segment dependence.

The key results are as follows. Cross-segment dependence is parsimoniously captured by a single latent process $\eta_t$, avoiding the need for explicit copula modeling or high-dimensional simulation. Conditional on $\eta_t$, segments are independent, reducing the route distribution from an $m$-dimensional to a one-dimensional problem. The Gamma moment-matching approximation for the conditional sum preserves the conjugate structure, yielding a closed-form $F$-distribution for the route predictive. Finally, the model updates sequentially in closed form, enabling real-time reliability tracking as new sensor data arrives.

The empirical application to I-55 demonstrates that the model achieves well-calibrated route-level predictions. With hyperparameters $(\alpha, \gamma) = (1.0, 0.70)$ selected via empirical Bayes and segment rate parameters $\lambda_j$ estimated via MCMC (Gibbs sampling with FFBS), the model produces predictive intervals with correct coverage and passes the marginal PIT calibration test. The posterior distributions of $\lambda_j$ concentrate tightly around their MLE values, and the Bayesian predictive is virtually identical to the plug-in predictive, confirming that the data are sufficiently informative for the conjugate online updating to be used directly in real-time operations. These results are particularly notable because the route predictive is obtained from the segment-level multivariate model---it is not a separate model fit to route data but rather a derived consequence of the conditional independence structure.

Our dynamic Gamma model relates to the GARCH framework adopted for travel time reliability forecasting \citep{Yang2018}: both address the time-varying volatility of travel times. In a GARCH model, the conditional variance $\sigma_t^2$ evolves as a function of past squared errors and past variances, producing confidence intervals that widen during congestion and narrow during free-flow. In our model, the discount factor $\gamma$ plays an analogous role to the GARCH persistence parameter: it controls how quickly past information decays and how responsive the model is to recent shocks. The key difference is that our model operates within a conjugate Bayesian framework that yields the full predictive \emph{distribution}---not just confidence intervals around a point forecast---and does so at the segment level, enabling composition to the route level through the common environment structure.

Several extensions merit future investigation. First, as noted in Section~\ref{sec:multivariate}, the single common environment $\eta_t$ implies equicorrelation across all segment pairs, whereas empirical correlation matrices exhibit distance decay. A natural extension is to introduce \emph{multiple} environment processes---for example, one per sub-corridor or one per congestion regime---each shared by a subset of segments. This would produce block-equicorrelation structures that better approximate the empirical spatial pattern, at the cost of additional latent processes. Similarly, allowing segment-specific shape parameters $\alpha_j$ (rather than a common $\alpha$) would better capture segment-level distributional heterogeneity, though at the cost of losing the exact conjugate structure. Second, the inclusion of observed covariates $u_t$ (weather, incidents, special events) through the $e^{\beta^T u_t}$ term in the model can improve prediction during non-recurrent conditions. Third, the Weibull extension suggested by \citet{APS20}, where $y_{jt} \mid \eta_t \sim \mathrm{Wei}(\lambda_j \eta_t, \phi_j)$, provides additional flexibility in the tail behavior at the cost of losing the closed-form route predictive. Fourth, allowing the segment parameters $\lambda_j$ to evolve dynamically can capture segment-specific non-stationarity such as construction zones or recurring bottleneck formation.

Fifth, the static mixture decomposition in Section~\ref{sec:application} (Figure~\ref{fig:mixture_decomposition}) reveals clear bimodality in route travel times, with distinct free-flow ($\mu = 19$ min) and congested ($\mu = 31$ min) components. The current dynamic model handles regime switching \emph{across time periods} through $\eta_t$, which shifts the predictive toward the appropriate mode as conditions evolve. However, within a single time period, the Gamma assumption is unimodal. A natural extension is a \emph{mixture of dynamic Gamma models}, $p(y_t \mid \eta_t) = \sum_{k=1}^K \pi_k \,\mathrm{Gam}(\alpha_k, \lambda_k \eta_t)$, where each component corresponds to a traffic regime. Component-specific parameters $(\alpha_k, \gamma_k)$ are essential: when all components share the same $(\alpha, \gamma, \lambda)$, the discount factor erases initialization differences at rate $\gamma^t$, causing the components to collapse to identical states. The common environment $\eta_t$ would still drive shared conditions across all regimes and segments, while the mixture captures within-period multi-modality. The closed-form route distribution would no longer hold directly, since the conditional sum would itself be a mixture of Gammas; variational or simulation-based approximations would be required for route-level inference. Inference would require augmenting the model with latent regime indicators, naturally handled by Gibbs sampling.

Finally, the model can be extended to network-level reliability by considering multiple routes that share common segments and correlated environment processes, enabling reliability-based network traffic management.

\bibliographystyle{apalike}
\bibliography{ref}

\end{document}